\documentclass{article}

\usepackage{graphicx}
\usepackage{xcolor}

\usepackage{float}
\usepackage{amsmath}
\usepackage{amsthm}
\usepackage{amssymb}
\usepackage{authblk}
\usepackage{url}
\usepackage[noadjust,nocompress]{cite}
\usepackage{hyperref}
\usepackage{cleveref}
\usepackage{thm-restate}
\usepackage{multirow}
\usepackage{fullpage}

\newtheorem{theorem}{Theorem}
\newtheorem{lemma}[theorem]{Lemma}

\newtheorem{property}[theorem]{Property}

\newcommand{\Z}{{\mathbb{Z}}}
\newcommand{\bin}{{\mathsf{bin}}}
\renewcommand{\int}{{\mathsf{int}}}
\newcommand{\lsb}{{\mathsf{LSB}}}

\newcommand{\LCP}{{\mathsf{LCP}}}

\newcommand{\DPre}{{\mathsf{DPre}}}
\newcommand{\Rev}{{\mathsf{Rev}}}
\newcommand{\Rot}{{\mathsf{Rot}}}

\newcommand{\Ratio}{{\mathsf{Ratio}}}

\newcommand{\id}{{\mathrm id}}

\newcommand{\calC}{{\mathcal{C}}}
\newcommand{\calH}{{\mathcal{H}}}
\newcommand{\calR}{{\mathcal{R}}}
\newcommand{\calX}{{\mathcal{X}}}
\newcommand{\calM}{{\mathcal{M}}}

\newcommand{\calE}{{\mathcal{E}}}
\newcommand{\calF}{{\mathcal{F}}}
\newcommand{\calT}{{\mathcal{T}}}
\newcommand{\depth}{{\mathsf{depth}}}
\newcommand{\height}{{\mathsf{height}}}
\newcommand{\cost}{{\mathsf{cost}}}
\newcommand{\tmax}{{t_{\max}}}
\DeclareMathOperator*{\argmin}{arg\,min}
\DeclareMathOperator*{\argmax}{arg\,max}
\DeclareMathOperator{\polylog}{polylog}

\title{Sensitivity and Size Relationships of the Lempel--Ziv Factorization}
\author[1]{Hiroki Shibata\thanks{\texttt{shibata.hiroki.753@s.kyushu-u.ac.jp}}}
\author[1]{Yuto Fujie\thanks{\texttt{fujie.yuto.104@s.kyushu-u.ac.jp}}}
\affil[1]{Joint Graduate School of Mathematics for Innovation, Kyushu University, Japan}

\begin{document}

\maketitle

\begin{abstract}
The Lempel--Ziv (LZ) factorization is a fundamental method for compressing repetitive strings.
Its multiplicative sensitivity to an operation measures how much that operation can increase the number of phrases, as the maximum ratio of the phrase counts after and before the operation.
While its sensitivity to single-character edits is known, its sensitivity to more global operations such as prefix deletion and string reversal remains unknown.
Further open questions concern its size relationships with other compressed representations and whether it can be converted into an LZ-like encoding of polylogarithmic height without asymptotically increasing its size.

In this paper, We determine the multiplicative sensitivity of the LZ factorization to prefix deletion, cyclic rotation, and string reversal by establishing an $\Omega(\log n)$ lower bound for each operation that matches the known upper bound, where $n$ is the input length.
Using the result for prefix deletion, we show that the combined size of the relative LZ factorization and the LZ encoding of its reference string can be asymptotically smaller than that of the LZ factorization.
We also obtain asymptotically tight worst-case ratios of $\Theta(\log n)$ for the size of the LZ factorization divided by the minimum size of a collage system and for the size of the lex-parse divided by that of the LZ factorization.

We give lower and upper bounds on the worst-case ratio of the minimum size of an LZ-like encoding of height at most $h$ to the size of the LZ factorization, with matching bounds for $h \geq \lg n$.
In particular, we show that the minimum size of an LZ-like encoding with height at most $h \in O(\polylog n)$ can be larger than the size of the LZ factorization by a factor of $\Omega(\log n / \log\log n)$.
All our lower bounds follow from a common string construction based on the bit-reversal permutation.
\end{abstract}

\section{Introduction}

A \emph{highly repetitive string} is a string in which most of the text can be expressed as copies of other parts of the string.
Such strings are common in large data collections, such as collections of genomic sequences and versioned documents, but their repetitive structure is not fully captured by traditional measures based on statistical properties.
\emph{Repetitiveness measures}~\cite{Navarro21a} were introduced to quantify the repetitive structure that these statistical measures do not fully capture.
These measures are either based on an actual compression method or defined directly from a combinatorial property of the string.

The \emph{Lempel--Ziv (LZ) factorization}~\cite{ZivL77} is one of the most fundamental methods for compressing strings.
It parses a string from left to right into \emph{phrases} by choosing the longest prefix of the remaining suffix that occurs earlier in the string at each step.
If no such prefix exists, the next phrase is a single character.
An \emph{LZ encoding} stores each phrase as either a single character or one source position together with the phrase length.
First introduced nearly fifty years ago, LZ encodings remain central to both theory and practice.
Because of their simplicity and compression performance, they are used in classical compression tools such as \texttt{gzip} and modern tools such as \texttt{zstd} and \texttt{Brotli}~\cite{rfc1952,rfc8878,AlakuijalaFFKOS19}.
In this sense, the number $z(T)$ of phrases in the LZ factorization of a string $T$ is also used as a central repetitiveness measure.
In this paper, we address three questions concerning the LZ factorization.

The first question concerns the \emph{sensitivity} of the size of the LZ factorization to operations that change a large part of a string.
Sensitivity to an edit operation measures the maximum increase in the value of a repetitiveness measure when the operation is applied to a string~\cite{AkagiFI23}.
Since sensitivity measures how much the compressed size can increase after an edit, it quantifies the robustness of the compression method.
While the sensitivity of the LZ factorization to single-character edits is fully determined~\cite{AkagiFI23,FujieSNI25}, its sensitivity to more global operations is still an open question.
In particular, the sensitivities to \emph{string reversal} and \emph{prefix deletion} have been studied in the context of the \emph{symmetry} and the \emph{monotonicity} of repetitiveness measures~\cite{BannaiFGINPU26,GiulianiILPST21,MitsuyaNIBT21,KociumakaNP23}.
Sensitivity to prefix deletion is also related to compression schemes based on \emph{relative LZ factorization} that construct a reference string $S$ from the input string $T$ and encode $T$ using references to substrings of $S$~\cite{KuruppuPZ10,KuruppuPZ11,KosolobovVNP20}.
Whether these sensitivities are bounded by a constant has remained open for many years.

The second question is about the \emph{size relationships} between the LZ factorization and two other compressed representations, collage systems and the lex-parse.
A \emph{collage system}~\cite{KidaMSTSA03} is a grammar with additional rules for repetition and substring truncation, and $c(T)$ denotes the minimum size of such a system deriving a string $T$.
The \emph{lex-parse}~\cite{NavarroOP21} is a variant of the LZ factorization that chooses matches according to the lexicographic order of suffixes, and $v(T)$ denotes its number of phrases for $T$.
While $c \in O(z)$ holds for every string and $v \in o(z)$ holds for some family of strings~\cite{NavarroOP21,Navarro21a}, it remained open whether there exists a family of strings with $c \in o(z)$ or one with $z \in o(v)$.

Finally, we consider the possibility of \emph{balancing} the LZ factorization.
In an LZ encoding, each text position in a copy phrase refers to an earlier occurrence of the same character.
These reference relations naturally induce a referencing forest whose height determines a bound on random access time.
An LZ-like encoding in which this forest has bounded height is called a \emph{height-bounded LZ-like encoding}~\cite{BannaiFHMP24}.
One of the central open problems for height-bounded LZ-like encodings is whether the LZ factorization of every string $T$ of length $n$ can be \emph{balanced}, yielding an LZ-like encoding with $O(z(T))$ phrases and a forest of height $O(\polylog n)$.
Such a balancing method would enable efficient random access using space linear in the size of the LZ factorization, resolving a long-standing open problem.

\paragraph*{Our Contributions}

\begin{table}[t]
\centering
\renewcommand{\arraystretch}{0.95}
\setlength{\tabcolsep}{6pt}
\begin{tabular}{@{}ccccc@{}}
\hline
\noalign{\vskip 3pt}
& & \multicolumn{2}{c}{Compared Quantities} & \\[3pt]
\cline{3-4}
\noalign{\vskip 3pt}
Worst-Case Ratio & Lower Bound & Denominator & Numerator & Reference \\[4pt]
\hline
\noalign{\vskip 5pt}
\multirow[c]{2}{*}{$\DPre_z(n)$} & \multirow[c]{2}{*}{$\Omega(\log n)$} & $z(B_\id B_\rho)$ & $z(B_\rho)$ & \multirow[c]{2}{*}{Theorem~\ref{thm:construction_prefix_deletion}} \\
 & & {\small (Theorem~\ref{thm:z_B_id_rho})} & {\small (Theorem~\ref{thm:z_B_rho})} & \\[8pt]
\multirow[c]{2}{*}{$\Rev_z(n)$} & \multirow[c]{2}{*}{$\Omega(\log n)$} & $z(B_\rho^R)$ & $z(B_\rho)$ & \multirow[c]{2}{*}{Theorem~\ref{thm:construction_string_reversal}} \\
 & & {\small (Theorem~\ref{thm:z_B_rho_R})} & {\small (Theorem~\ref{thm:z_B_rho})} & \\[8pt]
\multirow[c]{2}{*}{$\Ratio_{v,z}(n)$} & \multirow[c]{2}{*}{$\Omega(\log n)$} & $z(B_\id)$ & $v(B_\id)$ & \multirow[c]{2}{*}{Theorem~\ref{thm:construction_lex_lz_ratio}} \\
 & & {\small (Theorem~\ref{thm:z_B_id_rho})} & {\small (Theorem~\ref{thm:v_B_id})} & \\[8pt]
\multirow[c]{2}{*}{$\Ratio_{\hat{z}_h,z}(n)$} & \multirow[c]{2}{*}{$\Omega(Q(\lfloor\sqrt{n}\rfloor, h+1))$} & $z(B_\id)$ & $\hat{z}_h(B_\id)$ & \multirow[c]{2}{*}{Theorem~\ref{thm:construction_height_bounded_ratio}} \\
 & & {\small (Theorem~\ref{thm:z_B_id_rho})} & {\small (Theorem~\ref{thm:zh_B_id})} & \\[8pt]
\noalign{{\color{black!35}\hrule height 0.2pt}\vskip 5pt}
$\Rot_z(n)$ & $\Omega(\log n)$ & $z(B_\id B_\rho)$ & $z(B_\rho B_\id)$ & Corollary~\ref{cor:construction_cyclic_rotation} \\[8pt]
$\Ratio_{z,z_{\mathrm{rel}}}(n)$ & $\Omega(\log n)$ & $z_{\mathrm{rel}}(B_\rho)$ & $z(B_\rho)$ & Corollary~\ref{cor:construction_lz_rlz_ratio} \\[8pt]
$\Ratio_{z,c}(n)$ & $\Omega(\log n)$ & $c(B_\rho)$ & $z(B_\rho)$ & Corollary~\ref{cor:construction_lz_collage_ratio} \\[3pt]
\hline
\end{tabular}
\caption{
List of our lower bounds and the string instances used in their proofs.
For the first four results, we prove size bounds for the numerators and denominators.
The remaining results follow as corollaries from these size bounds.
Here $Q(r, h) = \min\{q \in \Z_{\geq 0} \mid \binom{h+q}{h} \geq r\}$ for positive integers $r$ and $h$.
The asymptotically tight bound for $Q(r, h)$ is shown in Theorem~\ref{thm:order_eval_Q}.
}
\label{tab:lower_bound_constructions}

\end{table}

In this paper, we resolve these three questions about the LZ factorization.
Our contributions are summarized as follows.
\begin{enumerate}
\item \textbf{Sensitivity of the LZ factorization.}
We determine the multiplicative sensitivity of the LZ factorization to prefix deletion, cyclic rotation, and string reversal.
For each operation, we give a family of strings in which a string of length $n$ has sensitivity $\Omega(\log n)$ to that operation.
General upper bounds of $O(\log n)$ show that all three bounds are asymptotically tight.
The construction for prefix deletion also gives pairs of strings $S$ and $T$ for which the total number of phrases in the LZ factorization of $S$ and the relative LZ factorization of $T$ with respect to $S$ is asymptotically smaller than $z(T)$.

\item \textbf{Size relationships among repetitiveness measures.}
We determine the size relationships among the LZ factorization, collage systems, and the lex-parse.
Specifically, we give a family of strings whose LZ factorizations are $\Omega(\log n)$ times larger than their minimum collage systems, and a family of strings whose lex-parses are $\Omega(\log n)$ times larger than their LZ factorizations.
General upper bounds of $O(\log n)$ on these ratios show that both worst-case ratios are asymptotically tight.

\item \textbf{Upper and lower bounds on the size of height-bounded LZ-like encodings.}
We establish a lower bound on the worst-case size ratio of LZ-like encodings of height at most $h$ to the LZ factorization, together with a matching upper bound for $h \geq \lg n$.
In particular, this lower bound shows that restricting the height to $O(\polylog n)$ can increase the number of phrases by a factor of $\Omega(\log n / \log\log n)$.
These results show the impossibility of \emph{balancing} the LZ factorization without asymptotically increasing its size.
\end{enumerate}

Remarkably, all our lower bounds are based on a common string construction using the \emph{bit-reversal permutation}.
All the lower-bound proofs use two strings $B_\id$ and $B_\rho$ obtained from this construction, together with their concatenations and reversals.
Table~\ref{tab:lower_bound_constructions} summarizes these lower bounds and the quantities compared in their proofs.

\subsection*{Related Works}

The notion of sensitivity for repetitiveness measures was first formalized by Akagi et al.~\cite{AkagiFI23}, and subsequent work has established tight bounds for basic edit operations.
In particular, tight bounds are now known for both the multiplicative and additive sensitivity of the LZ factorization to single-character substitutions, insertions, and deletions~\cite{AkagiFI23,FujieSNI25}.
For operations beyond single-character edits, Bathie et al.~\cite{BathieHLZ26} consider the sensitivity to $k$ substitutions and show that applying them to a string $T$ yields a string whose LZ factorization has $O(z(T) + k)$ phrases.

Sensitivity to string reversal has been actively studied, both before and after the general notion of sensitivity for repetitiveness measures was proposed~\cite{CohnH97,GiulianiILPST21}.
The number of runs in the Burrows--Wheeler transform (BWT) is a known example with large sensitivity to string reversal.
Before the general notion of sensitivity was proposed, Giuliani et al.~\cite{GiulianiILPST21} had already studied the worst-case ratio between the numbers of BWT runs before and after string reversal and proved an $O(\log^2 n)$ upper bound and an $\Omega(\log n)$ lower bound.
Recently, Bannai et al.~\cite{BannaiFGINPU26} investigated the sensitivity of many repetitiveness measures to string reversal.
Their work gives a lower bound of $3$ for the multiplicative sensitivity of the LZ factorization to string reversal, but whether this sensitivity is bounded by a constant remained unclear.
They summarize the question as follows: ``for LZ parsing and some of its variants, it has been conjectured that even though the number of phrases of the parses can change after reversing the string, the ratio $z(w^R)/z(w)$ is bounded by a constant''~\cite{BannaiFGINPU26}, where $w^R$ denotes the reversal of $w$.
Furthermore, it is known that disproving the conjecture would also yield a family of strings with $c \in o(z)$~\cite{NavarroOP21}.

As with string reversal, it was not known whether the multiplicative sensitivity of the LZ factorization to prefix deletion is bounded by a constant or not.
For the LZ factorization, monotonicity with respect to suffix deletion follows directly from the definition of the number of phrases, but no such monotonicity holds for prefix deletion.
Some previous studies have analyzed this sensitivity, and the best previously known lower bound is $3/2$, obtained from a construction for cyclic rotation by one character~\cite{MitsuyaNIBT21,BannaiCR24}.

The \emph{relative LZ (RLZ) factorization} represents an input string or a collection of strings as a sequence of phrases referring to occurrences in a string $S$ constructed from the input~\cite{KuruppuPZ10,KuruppuPZ11,KosolobovVNP20}.
Let $z(T \mid S)$ denote the number of phrases in the RLZ factorization of a string $T$ whose phrases refer only to occurrences in $S$.
Consider a pair of strings $S$ and $T$ for which $z(S) + z(T \mid S)$ is much smaller than $z(T)$.
In this case, storing $S$ using LZ compression and $T$ using RLZ factorization can save space compared with storing $T$ using LZ compression alone.
The existence of such a pair of strings is closely related to the sensitivity of the LZ factorization to prefix deletion.
Since every phrase in the RLZ factorization of $T$ refers to an occurrence in $S$ that precedes it in $ST$, we have $z(ST) \leq z(S) + z(T \mid S)$.
Deleting the prefix $S$ from $ST$ therefore increases the size of the LZ factorization from $z(ST)$ to the much larger value $z(T)$.
Compression methods based on RLZ have been studied both theoretically and experimentally~\cite{KuruppuPZ11,GagiePV16,Valenzuela16}, but it remained open whether $z(S) + z(T \mid S)$ could be asymptotically smaller than $z(T)$.

For the size relationship between the lex-parse and the LZ factorization, it is known that the lex-parses of \emph{Fibonacci strings} with even indices are asymptotically smaller than their LZ factorizations~\cite{NavarroOP21}.
However, it was not known whether the opposite direction holds.
Nakashima et al.~\cite{0001KFIB24} show that a change of alphabet order or a single-character edit can increase the size of the lex-parse by a factor of $\Theta(\log n)$, but even their examples do not give an instance for which the lex-parse is larger than the LZ factorization. 

A height-bounded LZ-like encoding~\cite{BannaiFHMP24} is an LZ-like encoding whose induced referencing forest has bounded height.
Another height-bounded variant of LZ-like encoding called BAT-LZ was also independently proposed with a slightly different definition~\cite{LiptakM024}.
For a string $T$ of length $n$ and a positive integer $h$, let $\hat{z}_h(T)$ denote the minimum number of phrases in an LZ-like encoding of $T$ with height at most $h$.
Bannai et al.~\cite{BannaiFHMP24} proved that $\hat{z}_{c\log n}(T)$ can be bounded by the size of the smallest run-length grammar~\cite{NishimotoIIBT16} up to a constant factor for some constant $c$ and is asymptotically smaller for some family of strings.
Cicalese and Ugazio~\cite{CicaleseU25} established an asymptotic separation between $\hat{z}_h(T)$ and $z(T)$ for constant $h$, but whether such a separation exists for polylogarithmic $h$ remained open.
Bannai et al.~\cite{BannaiFHMP24} conclude their paper with the question, ``Can we balance the height (achieve $O(\polylog(n))$ height) of an LZ-like encoding or a modified LZ-like encoding, by only increasing the size by a constant factor?''

While explicit upper bounds are not available in the literature for all of these problems, combining known results yields an $O(\log n)$ upper bound for each of them.
For string reversal, Bannai et al.~\cite{BannaiFGINPU26} give this upper bound explicitly.
For the remaining problems, we can easily obtain the upper bounds through the substring complexity.
Indeed, $\delta \leq z$ and $z \in O(\delta\log n)$~\cite{KociumakaNP23} imply an $O(\log n)$ upper bound on the sensitivity to prefix deletion and cyclic rotation since the substring complexity increases by at most a constant factor under these operations.
Since the minimum size of a collage system and the size of the lex-parse satisfy the same bounds in terms of the substring complexity~\cite{Navarro21a},
their size ratios with the LZ factorization are also bounded by $O(\log n)$.
The same argument gives an $O(\log n)$ upper bound on the ratio between the minimum size of an LZ-like encoding of height at most $c \log n$ for some constant $c$ and the size of the LZ factorization~\cite{Navarro21a,BannaiFHMP24}.
 \section{Preliminaries}

\subsection{Basic Notations}

For a positive integer $n$, let $\Z_{n} = \{0, \dots, n-1\}$.
For integers $a \leq b$, let $[a, b] = \{i \in \Z \mid a \leq i \leq b\}$.
We denote logarithms to bases $2$ and $e$ by $\lg$ and $\ln$, respectively.
For $x > 0$, let $\lg^+(x) = \max\{1, \lg x\}$.
Let $\Sigma$ be an alphabet.
An element of $\Sigma$ is called a character.
A string $T$ of length $n = |T|$ over the alphabet $\Sigma$ is a sequence $T[0] \cdots T[n-1]$ of characters such that $T[i] \in \Sigma$ for all $i \in \Z_{n}$.
Let $\varepsilon$ denote the string of length $0$.
For any two strings $X$ and $Y$, we denote by $XY = X[0] \cdots X[|X|-1] Y[0] \cdots Y[|Y|-1]$ the concatenation of $X$ and $Y$.
For any string $X$, we denote by $X^R = X[|X|-1] \cdots X[0]$ its reversal.
If $T = XYZ$ for some strings $X, Y, Z \in \Sigma^*$, then $X$, $Y$, and $Z$ are called a prefix, a substring, and a suffix of $T$, respectively.
For $0 \leq i \leq j < n$, we denote by $T[i..j] = T[i]\cdots T[j]$ the substring of $T$ from position $i$ to $j$.
For convenience, define $T[i..j] = \varepsilon$ when $j < i$.
For any strings $X$ and $Y$, let $\LCP(X, Y)$ denote the length of the longest common prefix of $X$ and $Y$.

Let $\prec$ be a total order on $\Sigma$.
The \emph{lexicographic order} on $\Sigma^*$ defined by $\prec$ is given by $X \prec Y$ if and only if $X$ is a proper prefix of $Y$ or there exists an integer $0 \leq i < \min\{|X|,|Y|\}$ such that $X[0..i-1]=Y[0..i-1]$ and $X[i] \prec Y[i]$ for any strings $X, Y \in \Sigma^*$.

We define the binary alphabet $\Sigma_2 = \{\mathtt{0}, \mathtt{1}\}$.
For an integer $x \geq 0$, let $\bin(x) \in \left(\Sigma_2\right)^*$ be the binary representation of $x$.
For integers $k \geq 0$ and $x \in \Z_{2^k}$, let $\bin_k(x)$ be the $k$-bit binary representation of $x$.
For a binary string $B$, let $\int(B)$ be the integer $x$ such that $\bin_{|B|}(x) = B$.
For an integer $x \geq 1$, let $\lsb(x)$ be the position of the least significant $\mathtt{1}$ in $\bin(x)$, where the rightmost bit is at position $0$.
Equivalently, $\bin(x)$ has $\mathtt{1}\mathtt{0}^{\lsb(x)}$ as a suffix.
For example, $\bin_4(4) = \texttt{0100}$, $\int(\mathtt{0100}) = 4$, and $\lsb(4) = 2$.

\subsection{Lempel--Ziv Factorization}
A sequence of nonempty strings $F_1, \dots, F_m$ is a \emph{factorization} of a string $T$ if $T = F_1 \cdots F_m$.
Each element $F_i$ is called a \emph{phrase}, and its starting position in $T$ is $b_i = \sum_{k=1}^{i-1} |F_k|$.
An \emph{LZ-like factorization}~\cite{BannaiFHMP24} of $T$ is a factorization in which every phrase $F_i$ with $|F_i| > 1$ satisfies $F_i = T[j..j + |F_i| - 1]$ for some $0 \leq j < b_i$.
The \emph{Lempel--Ziv (LZ) factorization}~\cite{ZivL77} is the LZ-like factorization obtained by choosing a longest match at each step from left to right.
Specifically, $F_i$ is the longest prefix of $T[b_i..|T|-1]$ with an occurrence starting before $b_i$, or $F_i = T[b_i]$ if no such nonempty prefix exists.
This greedy factorization has the minimum number of phrases among all LZ-like factorizations of $T$~\cite{BannaiFHMP24}.
We denote the size of the LZ factorization of $T$ by $z(T)$.
The measure $z$ is monotone under appending a suffix.
That is, $z(X) \leq z(XY)$ for any strings $X$ and $Y$.

We represent each phrase in an LZ-like factorization as either a \emph{literal phrase} of length $1$ or a \emph{copy phrase} with an earlier occurrence.
For each copy phrase $F_i$, we choose a \emph{source position} $s_i$ such that $0 \leq s_i < b_i$ and $F_i = T[s_i..s_i + |F_i| - 1]$.
For each literal phrase $F_i$, we set $s_i = \bot$.
A phrase of length $1$ with an earlier occurrence can be represented as either a literal phrase or a copy phrase.
An \emph{LZ-like encoding} of $T$ is a pair $\calE = (\calF, (s_i)_{i=1}^m)$ consisting of an LZ-like factorization $\calF = (F_1, \dots, F_m)$ and a sequence of source positions $(s_i)_{i=1}^m$.
The size of an LZ-like encoding is the number of phrases in its underlying factorization.
In both LZ and LZ-like encodings, the occurrence starting at the selected source position may overlap the corresponding phrase.

Let $\calE = (\calF, (s_j)_{j=1}^f)$ be an LZ-like encoding of a string $T$ of length $n$ with underlying factorization $\calF = (F_1, \dots, F_f)$.
The \emph{transition function} $\tau_{\calE}: \Z_n \cup \{ \bot \} \to \Z_n \cup \{ \bot \}$ of $\calE$ is defined by $\tau_{\calE}(\bot) = \bot$ and
\[
\tau_{\calE}(i) =
\begin{cases}
\bot & \text{if $s_j = \bot$}, \\
s_j + ((i - b_j) \bmod (b_j - s_j)) & \text{otherwise}
\end{cases}
\]
for each position $b_j \leq i < b_j + |F_j|$.
For every position $i$ contained in a copy phrase, we have $\tau_{\calE}(i) < i$ and $T[\tau_{\calE}(i)] = T[i]$.
We define $\tau_{\calE}$ so that each position in $F_j$ refers to a position in a different phrase when the interval of $F_j$ overlaps its source interval.

The edges between $i$ and $\tau_{\calE}(i)$ for all positions with $\tau_{\calE}(i) \neq \bot$ form a rooted forest on $\Z_n$.
The roots of this forest correspond to the positions $i$ satisfying $\tau_{\calE}(i) = \bot$, namely the positions of the literal phrases.
The \emph{depth} of a position $i \in \Z_n$ is its depth in this forest.
Equivalently, $\depth_{\calE}(i) = \min \{ k \in \mathbb{Z}_{\geq 0} \mid \tau_{\calE}^{k + 1}(i) = \bot \}$.
The \emph{height} of an LZ-like encoding $\calE$ is defined as $\height(\calE) = \max_{i \in \Z_n} \depth_{\calE}(i)$.
For any positive integer $h$ and a string $T$, let $\hat{z}_{h}(T)$ be the size of the smallest LZ-like encoding whose height is at most $h$.
Since every LZ-like encoding has height at most $n$, $\hat{z}_{n}(T) = z(T)$ holds for every string $T$.

\subsection{Other Repetitiveness Measures and Compression Methods}
\label{sse:other_measures}

A \emph{repetitiveness measure}~\cite{Navarro21a} is a function $f: \Sigma^* \rightarrow \mathbb{R}_{\geq 0}$ that measures repetition in a string.
For two repetitiveness measures $f$ and $g$ and a positive integer $n$, the \emph{worst-case ratio} of $f$ to $g$ is defined by $\Ratio_{f,g}(n) = \max_{T \in \Sigma^n} f(T) / g(T)$.

\paragraph*{Relative LZ Factorization}
For strings $T$ and $S$, the \emph{relative LZ (RLZ) factorization}~\cite{KuruppuPZ10,KuruppuPZ11} of $T$ with respect to $S$ is the factorization $F_1, \dots, F_k$ obtained from left to right by choosing the longest prefix of the remaining suffix of $T$ that occurs as a substring of $S$.
If no such nonempty prefix exists, the next character forms a literal phrase.
Every copy phrase refers to an occurrence entirely contained in $S$.
We denote the number of phrases in this factorization by $z(T \mid S)$.
We define the minimal total size of a relative LZ factorization of $T$ as $z_{\mathrm{rel}}(T) = \min_{S \in \Sigma^*} \{z(S) + z(T \mid S)\}$.
Here, the total size is the sum of the sizes of the LZ factorization of $S$ and the RLZ factorization of $T$ with respect to $S$.
Since we can construct an LZ-like factorization of $ST$ by concatenating the LZ factorization of $S$ and the RLZ factorization of $T$ with respect to $S$, we have $z(ST) \leq z(S) + z(T \mid S)$.

\paragraph*{Lex-Parse}
The \emph{lex-parse}~\cite{NavarroOP21} is a variant of the LZ factorization in which copy phrases refer to lexicographically smaller suffixes.
For a string $T$ of length $n$, the lex-parse of $T$ with respect to $\prec$ is the factorization $F_1, \cdots, F_m$ of $T$ constructed from left to right as follows.
Let $\ell_i$ be the maximum of $\LCP(T[b_i..n-1], T[p..n-1])$ over all positions $p \in \Z_{n}$ satisfying $T[p..n-1] \prec T[b_i..n-1]$.
If there is no such position $p$, then $\ell_i = 0$.
The factor $F_i$ is defined as $F_i = T[b_i..b_i+\max\{\ell_i, 1\}-1]$.
We also define an encoding of the lex-parse of $T$ as a pair $\calE = (\calF, (s_i)_{i=1}^m)$ consisting of the lex-parse $\calF = (F_1, \dots, F_m)$ and a sequence of source positions $(s_i)_{i=1}^m$ satisfying $T[s_i..n-1] \prec T[b_i..n-1]$ for every copy phrase $F_i$.
For the lex-parse and its encoding, we use the same terms as for LZ factorizations and encodings, such as copy phrase, literal phrase, and source position.
We denote the number of phrases in the lex-parse of $T$ by $v_\prec(T)$, or simply by $v(T)$ when the order $\prec$ is clear.

\paragraph*{Collage Systems}
A \emph{collage system}~\cite{KidaMSTSA03} is an extension of context-free grammars and consists of nonterminals $X_1, \ldots, X_m$ with exactly one production rule for each nonterminal.
In addition to terminal rules $X_i \rightarrow c$ for $c \in \Sigma$ and concatenation rules $X_i \rightarrow X_jX_k$, it allows exponent rules $X_i \rightarrow (X_j)^r$ for $r \geq 2$ and truncation rules $X_i \rightarrow X_j[b..e]$.
The rule for $X_i$ refers only to nonterminals with smaller indices.
A collage system \emph{derives} the string derived by its last nonterminal $X_m$.
We define the size of a collage system as the number $m$ of its nonterminals.
For a string $T$, let $c(T)$ denote the minimum size of collage systems that derive $T$.
For every string $T$, $c(T) \in O(z(T))$ holds~\cite{NavarroOP21}.
The measure $c$ is symmetric, so $c(T) = c(T^R)$ holds for every string $T$.

\begingroup\emergencystretch=2em
\paragraph*{Substring Complexity}
The \emph{substring complexity} $\delta(T)$~\cite{KociumakaNP23,ChristiansenEKN21} of a string $T$
is defined as $\delta(T) = \allowbreak \max_{k \geq 1} d_k(T) / k$,
where $d_k(T)$ is the number of distinct substrings of length $k$ in $T$.
The substring complexity $\delta$ is monotone and symmetric.
That is, $\delta(X) \leq \delta(Y)$ for any strings $X$ and $Y$ such that $X$ is a substring of $Y$, and $\delta(X) = \delta(X^R)$ for every string $X$.
For any two strings $U$ and $V$, $\delta(UV) \leq \delta(U) + \delta(V) + 1$ holds.
Indeed, concatenating $U$ and $V$ produces at most $k-1$ new substrings of length $k$ across their boundary, so $d_k(UV) \leq d_k(U) + d_k(V) + k-1$ for every $k \geq 1$.
Dividing by $k$ and taking the maximum gives the inequality.
\par\endgroup

The following bounds follow from known results~\cite{Navarro21a} and basic properties of substring complexity.
The full proof is given in Appendix~\ref{sec:omitted_sensitivity_ratio_proofs}.
\begin{restatable}{theorem}{substringcomplexitybounds} \label{thm:substring_complexity_bounds}
For every nonempty string $T$ and every $f \in \{z, c, v, z_{\mathrm{rel}}\}$,
we have
$\delta(T) \in O(f(T))$
and
$f(T) \in O\left(\delta(T) \log\left(1 + \frac{|T|}{\delta(T)}\right)\right)$.
Consequently, $\Ratio_{f,g}(n) \in O(\log n)$ for any $f, g \in \{z, c, v, z_{\mathrm{rel}}\}$.
\end{restatable}

\subsection{Sensitivities}

The multiplicative \emph{sensitivity} of a repetitiveness measure is the maximum factor by which its value can increase when an operation is applied to a string~\cite{AkagiFI23}.
Each sensitivity considered below is a function of the string length $n$.
For each positive integer $n$ and each repetitiveness measure $f$, the sensitivities of $f$ to
\emph{prefix deletion} $\DPre_f(n)$,
\emph{string reversal} $\Rev_f(n)$, and
\emph{cyclic rotation} $\Rot_f(n)$ are defined as follows:
\begin{align*}
    \DPre_f(n) &= \max_{\substack{S,T \in \Sigma^*\\ |ST|=n}} f(T) / f(ST), \\
    \Rev_f(n)  &= \max_{\substack{T \in \Sigma^*\\ |T|=n}} f(T^R) / f(T), \\
    \Rot_f(n)  &= \max_{\substack{S,T \in \Sigma^*\\ |ST|=n}} f(TS) / f(ST).
\end{align*}
The following upper bounds on the sensitivities of $z$ follow from properties of substring complexity.
The full proof is given in Appendix~\ref{sec:omitted_sensitivity_ratio_proofs}.
\begin{restatable}{theorem}{sensitivityupperbounds} \label{thm:upper_bounds}
For each function $\textsf{Func} \in \{ \DPre, \Rev, \Rot \}$,
$\textsf{Func}_z(n) \in O(\log n)$.
\end{restatable}
 \section{Construction and Proof Overview}

In this section, we introduce the strings used in our lower-bound proofs, outline the proofs, and establish the common lemmas.

\subsection{String Construction}

For each positive integer $k$, the \emph{bit-reversal permutation} $\pi_k: \Z_{2^k} \to \Z_{2^k}$ is the bijection defined by $\pi_k(x) = \int\left(\bin_k(x)^R\right)$ for every $x \in \Z_{2^k}$.
For example, $\pi_4(5) = \int\left(\bin_4(5)^R\right) = \int\left((\mathtt{0101})^R\right) = \int\left(\mathtt{1010}\right) = 10$.
If we reverse the bits of a binary string of length $k$ twice, we obtain the original string.
Thus, $\pi_k(\pi_k(x)) = x$ for every $x \in \Z_{2^k}$ and $\pi_k = \pi^{-1}_k$.
For instance, $(\pi_3(0), \dots, \pi_3(7)) = (\pi^{-1}_3(0), \dots, \pi^{-1}_3(7)) = (0, 4, 2, 6, 1, 5, 3, 7)$.

For a fixed integer $\alpha \geq 3$, let $m = 2^\alpha$.
We define the alphabet $\Sigma_{\alpha}$ by $\Sigma_{\alpha} = \{ \mathtt{a}_i, \mathtt{b}_i, \mathtt{c}_i \mid i \in \mathbb{Z}_m\} \cup \{ \mathtt{\$}_x \mid x \in \mathbb{Z}_m \setminus \{ 0 \} \}$.
For each $x \in \Z_m \setminus \{0\}$, let $B_{\alpha,x}$ be a string of length $m$ with its $i$-th character $B_{\alpha,x}[i]$ defined as follows.
\[
B_{\alpha,x}[i] =
\begin{cases}
    \mathtt{a}_i & (\pi_\alpha(i) < x) \\
    \mathtt{b}_i & (\pi_\alpha(i) = x) \\
    \mathtt{c}_i & (\pi_\alpha(i) > x)
\end{cases}
\]
We call each string $B_{\alpha,x}$ a \emph{block} of order $\alpha$.

Our lower bounds use two strings obtained by concatenating blocks in different orders, with the delimiter $\mathtt{\$}_x$ after each block $B_{\alpha,x}$ in both strings.
For the first string, we concatenate the blocks in increasing order of their indices.
For the second string, we define the \emph{LSB-decreasing permutation} $\rho_\alpha: \Z_m \setminus \{0\} \to \Z_m \setminus \{0\}$ so that $\rho_\alpha(1), \ldots, \rho_\alpha(m-1)$ are in nonincreasing order of their $\lsb$ values.
If two integers have the same $\lsb$ value, they appear in increasing numerical order.
For example, $(\rho_3(1), \ldots, \rho_3(7)) = (4, 2, 6, 1, 3, 5, 7)$ and their corresponding $\lsb$ values are $(2, 1, 1, 0, 0, 0, 0)$.

Using these two orders, we define the strings $B_{\alpha,\id}$ and $B_{\alpha,\rho_\alpha}$ by
\begin{align*}
B_{\alpha, \id} &= B_{\alpha,1} \mathtt{\$}_1 \cdots B_{\alpha,m-1} \mathtt{\$}_{m-1}, \\
B_{\alpha, \rho_\alpha} &= B_{\alpha,\rho_\alpha(1)} \mathtt{\$}_{\rho_\alpha(1)} \cdots B_{\alpha,\rho_\alpha(m-1)} \mathtt{\$}_{\rho_\alpha(m-1)}.
\end{align*}

Both strings have length $(m-1)(m+1)=m^2-1$.
In what follows, we fix an arbitrary integer $\alpha \geq 3$ and set $m = 2^\alpha$.
We omit $\alpha$ from subscripts when its value is clear from the context.
For example, we write $B_x$ and $B_\rho$ for $B_{\alpha,x}$ and $B_{\alpha,\rho_\alpha}$, respectively.

Figure~\ref{fig:block_orders_m16} illustrates the blocks $B_x$ and their orders in $B_\id$ and $B_\rho$ for $m = 16$.
\begin{figure}[t]
\centering
\includegraphics{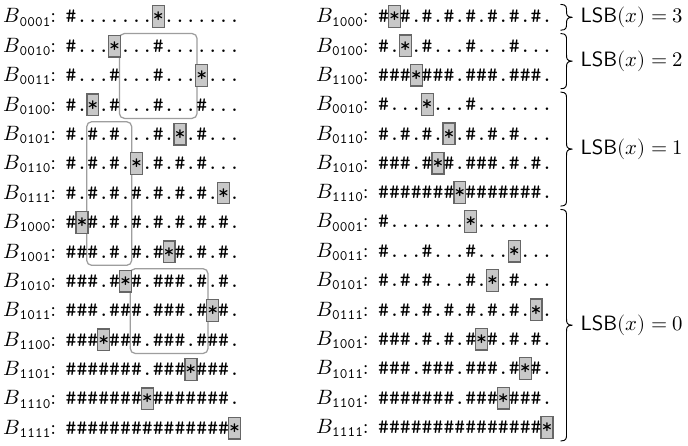}
\caption{
  The blocks $B_x$ and their orders in $B_\id$ and $B_\rho$ for $m = 16$.
  To make the block contents easier to see, we display $\mathtt{a}_i$, $\mathtt{b}_i$, and $\mathtt{c}_i$ as \texttt{\#}, \texttt{*}, and \texttt{.}, respectively.
  For each $x \in \Z_m \setminus \{0\}$, the gray box marks the character $\mathtt{b}_p$ at position $p = \pi_\alpha(x)$ in $B_x$.
  The block indices are written in binary.
  The blocks in the left and right panels appear from top to bottom in their order in $B_\id$ and $B_\rho$, respectively.
  The rounded rectangles in the left panel show examples of occurrences of the same substring at the same offsets in blocks with consecutive indices.
}
\label{fig:block_orders_m16}
\end{figure}

\subsection{Overview of the Lower-Bound Proofs}

We prove the first four results in Table~\ref{tab:lower_bound_constructions} and then derive the remaining results as corollaries.
For each of these four results, we bound the denominator by $O(m)$.
The corresponding numerators are $\Omega(m\alpha) = \Omega(m\log n)$ for the first three results and $\Omega(mQ(m, h+1))$ for the fourth, where $Q(r, h) = \min\{q \in \Z_{\geq 0} \mid \binom{h+q}{h} \geq r\}$ for positive integers $r$ and $h$ denotes the discrete inverse of the binomial coefficient $\binom{h+q}{h}$ with respect to $q$.
Its asymptotic order is given in Theorem~\ref{thm:order_eval_Q}.

To prove the first four results, we mainly use the following three lemmas.
The block-distance lemma (Lemma~\ref{lem:construction_block_lcs_bound}) bounds the length of a copy phrase in terms of the distance between its block and the source block.
Together with the total block-distance lemma (Lemma~\ref{lem:block_distance_sum_bound}), it yields the lower bounds on the numerators.
The extension lemma (Lemma~\ref{lem:extension_lemma}) then extends the resulting lower bounds to every sufficiently large string length $n$.
All three lemmas are stated in the next subsection.

In the remainder of the paper, we provide only proof sketches for some lemmas and theorems, while their complete proofs are given in Appendix~\ref{sec:omitted_sensitivity_ratio_proofs} and Appendix~\ref{sec:increasing_tree_cost_bounds}.

\subsection{Basic Properties and Common Lemmas}

We now establish the basic properties and common lemmas used in these proofs.

\begin{property} \label{pro:copy_phrases_in_blocks}
In each of $B_\id$ and $B_\rho$, any two occurrences of the same nonempty substring at different starting positions lie within different blocks and start at the same offset.
Namely, these occurrences are $B_x[l..r]$ and $B_y[l..r]$ for some distinct $x, y \in \Z_m \setminus \{0\}$ and integers $0 \leq l \leq r < m$.
\end{property}
\begin{property} \label{pro:basic_block_prop}
For any two strings $B_{\alpha,x}$ and $B_{\alpha,y}$ with $x < y$, their $i$-th characters are different if and only if $x \leq \pi_\alpha(i) \leq y$.
\end{property}
The first property is implied by the uniqueness of each delimiter and the fact that the same character can occur only at the same offset in different blocks.
The second property follows directly from the definition of the blocks.

The following \emph{block-distance lemma} is central to the proofs of our lower bounds.
\begin{restatable}[Block-distance lemma]{lemma}{blockdistancelemma} \label{lem:construction_block_lcs_bound}
Let $x, y \in \Z_m \setminus \{0\}$, $d = |y - x|$, and let $\ell$ be the length of a longest common substring between $B_{\alpha,x}$ and $B_{\alpha,y}$.
Then, $(d+1)\ell \leq 4m$.
\end{restatable}
\begin{proof}[Proof (sketch)]
The case $x = y$ is immediate.
Assume $x < y$ without loss of generality.
By Property~\ref{pro:copy_phrases_in_blocks}, the two occurrences of a common substring of length $\ell$ in blocks $B_{\alpha,x}$ and $B_{\alpha,y}$ start at the same offset, so $B_{\alpha,x}[l..r] = B_{\alpha,y}[l..r]$ for some $l, r$.
We take integers $a, q \geq 0$ such that $I = [a2^q, (a+1)2^q-1] \subseteq [x, y]$ and $2^{q+1} \leq d+1 < 2^{q+2}$.
The lower $q$ bits of the integers in $I$ take all possible values, so the set $\{\pi_\alpha(i) \mid i \in I\}$ contains integers with all possible values in their upper $q$ bits.
Thus, every interval of $2^{\alpha-q}$ consecutive positions in a block contains $\pi_\alpha(i)$ for some $i \in I$.
By Property~\ref{pro:basic_block_prop}, the two blocks differ at every such position, so $\ell < 2^{\alpha-q}$.
The length bound and the choice of $q$ yield the claimed inequality.
\end{proof}

The left-hand side of the inequality $(d+1)\ell \leq 4m$ in the above lemma corresponds to the total length of occurrences of the same substring at the same offset in blocks with consecutive indices.
Each rounded rectangle in Figure~\ref{fig:block_orders_m16} shows examples of these occurrences, and no such rectangle contains more than $4m$ characters.
Intuitively, this is because the bit-reversal permutation maps consecutive integers to widely distributed values.
By Property~\ref{pro:basic_block_prop}, two blocks $B_{\alpha,x}$ and $B_{\alpha,y}$ with $x < y$ differ exactly at positions $p$ such that $x \leq \pi_\alpha(p) \leq y$.
Since these positions are widely distributed, the blocks cannot have a long common substring.
In what follows, we mainly use the weaker bound $d\ell \leq 4m$.

Let $\calE$ be an LZ-like encoding or an encoding of the lex-parse of $B_\id$ or $B_\rho$.
We define the \emph{block-distance function} $d_\calE$ on $(\Z_m \setminus \{0\}) \times \Z_m$ by setting $d_\calE(x, i) = |x - y|$ for $l \leq i \leq r$ whenever a copy phrase $B_x[l..r]$ refers to the source occurrence $B_y[l..r]$.
We set $d_\calE(x, i) = 0$ if $B_x[i]$ is represented by a literal phrase.
By the block-distance lemma, we can bound the sum of the values of the block-distance function in terms of $|\calE|$.
\begin{lemma}[Total block-distance lemma] \label{lem:block_distance_sum_bound}
Let $\calE$ be an LZ-like encoding or an encoding of the lex-parse of $B_\id$ or $B_\rho$.
Then,
\[
\sum_{x=1}^{m-1} \sum_{i=0}^{m-1} d_\calE(x, i) \leq 4m|\calE|.
\]
\end{lemma}
\begin{proof}
By Property~\ref{pro:copy_phrases_in_blocks}, each copy phrase has the form $B_x[l..r]$ with source occurrence $B_y[l..r]$, where $x \neq y$.
For this phrase, Lemma~\ref{lem:construction_block_lcs_bound} gives $\sum_{i=l}^{r} d_\calE(x, i) = |x - y|(r-l+1) \leq 4m$.
Since there are at most $|\calE|$ copy phrases, the total sum is at most $4m|\calE|$.
\end{proof}

Since $B_\id$ and $B_\rho$ have only specific lengths depending on $\alpha$, we use the following \emph{extension lemma} to extend the resulting lower bounds to arbitrary string lengths.
\begin{restatable}[Extension lemma]{lemma}{extensionlemma} \label{lem:extension_lemma}
For any string $T$ of length $n$ over an ordered alphabet and any integer $n' \geq n + 2$, there exists a string $T'$ of length $n'$ such that $z(T') = z(T) + 2$, $z\!\left((T')^R\right) = z(T^R) + 2$, $v(T') = v(T) + 2$, and $z_{\mathrm{rel}}(T') \leq z_{\mathrm{rel}}(T) + 3$.
Moreover, $\hat{z}_h(T') = \hat{z}_h(T) + 2$ for every integer $h \geq 1$.
\end{restatable}
\begin{proof}[Proof (sketch)]
The string $T' = \mathtt{\#}^{n'-n}T$ satisfies the required conditions when $\mathtt{\#}$ is a new character lexicographically larger than every character in $T$.
The bound for $z_{\mathrm{rel}}(T')$ follows by using $\mathtt{\#}^{n'-n}S$ as the reference string for $T'$, where $S$ is a string satisfying $z_{\mathrm{rel}}(T) = z(S) + z(T \mid S)$.
\end{proof}
 \section{Sensitivity of the LZ Factorization}

\subsection{Sensitivity to Prefix Deletion and Cyclic Rotation}
\label{sse:pref_deletion}

We first study the sensitivity of the LZ factorization to prefix deletion.
To prove the lower bound, we use the string $B_\id B_\rho$ and the string $B_\rho$ obtained by deleting its prefix $B_\id$, and consider the ratio $z(B_\rho) / z(B_\id B_\rho)$.
The next two theorems give an upper bound on its denominator and a lower bound on its numerator, respectively.

\begin{theorem} \label{thm:z_B_id_rho}
$z(B_\id) \in O(m)$, $z(B_\rho \mid B_\id) \in O(m)$, and $z(B_\id B_\rho) \in O(m)$.
\end{theorem}
\begin{proof}
We first construct an LZ-like factorization of $B_\id$ with at most $7m - 11$ phrases and then append at most $m - 1$ phrases representing $B_\rho$.

The first substring $B_1\mathtt{\$}_1$ can be factorized into $m + 1$ literal phrases.
$B_x$ and $B_{x+1}$ differ only at positions $\pi(x)$ and $\pi(x+1)$ for every $1 \leq x < m-1$.
Hence, $B_{x+1}\mathtt{\$}_{x+1}$ can be factorized into at most six phrases by copying at most three substrings that are identical to the corresponding substrings of $B_x$ and using a literal phrase at each of these two positions and for $\mathtt{\$}_{x+1}$.
Consequently, $B_\id$ has an LZ-like factorization with at most $(m + 1) + 6(m - 2) = 7m - 11$ phrases, so $z(B_\id) \leq 7m - 11$.

Each substring $B_{\rho(k)}\mathtt{\$}_{\rho(k)}$ in $B_\rho$ can be represented by a single copy phrase whose source lies entirely within $B_\id$, so $z(B_\rho \mid B_\id) \leq m - 1$.
Concatenating the factorization of $B_\id$ and these phrases gives an LZ-like factorization of $B_\id B_\rho$ with at most $(7m - 11) + (m - 1) = 8m - 12$ phrases, so $z(B_\id B_\rho) \leq 8m - 12$.
\end{proof}

\begin{theorem} \label{thm:z_B_rho}
$z(B_\rho) \in \Omega(m\alpha)$.
\end{theorem}
\begin{proof}
We bound the number of phrases in each block by bounding their lengths.

Fix an integer $1 \leq k < m$, and let $x = \rho(k)$ and $q = \lsb(x)$.
For a copy phrase $B_x[l..r]$ of length $\ell = r-l+1$, Property~\ref{pro:copy_phrases_in_blocks} shows that its source occurrence is $B_y[l..r]$ in a different block $B_y$.
Let $k' = \rho^{-1}(y)$.
Since the source occurrence starts before the phrase, $1 \leq k' < k$.
Let $d = |x - y|$.
By the definition of $\rho$ and $k' < k$,
we have $\lsb(y) \geq \lsb(x)$,
hence both $x$ and $y$ are multiples of $2^q$.
Since $x \neq y$, we obtain $d \geq 2^q$.
Lemma~\ref{lem:construction_block_lcs_bound} now gives $\ell \leq 4m/d \leq 4m/2^q$.
Thus, every phrase in this block has length at most $4m/2^q$, and the block contains at least $\max\{1, m/(4m/2^q)\} = \max\{1, 2^q/4\}$ phrases.

For each integer $0 \leq q < \alpha$, exactly $2^{\alpha-q-1}$ block indices $x \in [1, m-1]$ satisfy $\lsb(x) = q$.
Thus, 
$z(B_\rho)
\geq \sum_{q=0}^{\alpha-1} 2^{\alpha-q-1}\max\{1, 2^q/4\}
\geq \sum_{q=0}^{\alpha-1}\frac{m}{8}
= \frac{m\alpha}{8}$, proving $z(B_\rho) \in \Omega(m\alpha)$.
\end{proof}

Together with the extension lemma, $z(B_\rho) \in \Omega(m\alpha)$ (Theorem~\ref{thm:z_B_rho}) and $z(B_\id B_\rho) \in O(m)$ (Theorem~\ref{thm:z_B_id_rho}) give $\DPre_z(n) \in \Omega(\log n)$.
Combining this lower bound with the upper bound in Theorem~\ref{thm:upper_bounds} gives the following result.
\begin{restatable}[Sensitivity to Prefix Deletion]{theorem}{prefixdeletionsensitivity} \label{thm:construction_prefix_deletion}
$\DPre_z(n) \in \Theta(\log n)$.
\end{restatable}

Since $z(X) \leq z(XY)$ for any pair of strings $X$ and $Y$, we have $\DPre_z(n) \leq \Rot_z(n)$.
$\DPre_z(n) \in \Theta(\log n)$ and the upper bound in Theorem~\ref{thm:upper_bounds} imply the following result.
\begin{restatable}[Sensitivity to Cyclic Rotation]{corollary}{cyclicrotationsensitivity} \label{cor:construction_cyclic_rotation}
$\Rot_z(n) \in \Theta(\log n)$.
\end{restatable}

Theorem~\ref{thm:z_B_id_rho} and Theorem~\ref{thm:z_B_rho}, together with the extension lemma (Lemma~\ref{lem:extension_lemma}) and the general upper bound in Theorem~\ref{thm:substring_complexity_bounds}, imply the following result.
\begin{restatable}[Worst-Case Ratio of the LZ Factorization to Relative LZ]{corollary}{lzrelativeratio} \label{cor:construction_lz_rlz_ratio}
$\Ratio_{z,z_{\mathrm{rel}}}(n) \in \Theta(\log n)$.
\end{restatable}

\subsection{Sensitivity to String Reversal}

We next study the sensitivity of the LZ factorization to string reversal using $B_\rho$ and its reverse $B_\rho^R$.
Theorem~\ref{thm:z_B_rho} already gives $z(B_\rho) \in \Omega(m\alpha)$.
We now show that $z(B_\rho^R) \in O(m)$.

\begin{restatable}{theorem}{reversedrhosize} \label{thm:z_B_rho_R}
$z(B_\rho^R) \in O(m)$.
\end{restatable}
\begin{proof}[Proof (sketch)]
We construct a variant of an LZ-like factorization of $B_\rho$ with $O(m)$ phrases, where each copy phrase refers to a later occurrence.

For each $1 \leq x < m-1$, we construct a factorization of $B_x$ so that its copy phrases refer to occurrences in $B_{x+2}$ if $x$ is odd and in $B_{x+1}$ if $x$ is even.
By the definition of LSB-decreasing order, this block occurs after $B_x$ in $B_\rho$.
Since the two blocks differ at only a constant number of positions, $B_x$ can be represented by a constant number of phrases.
We represent the last block and the delimiters by literal phrases.
Reversing the phrase order and the corresponding references gives an LZ-like factorization of $B_\rho^R$ with $O(m)$ phrases.
\end{proof}

By combining $z(B_\rho) \in \Omega(m\alpha)$ (Theorem~\ref{thm:z_B_rho}) and $z(B_\rho^R) \in O(m)$ (Theorem~\ref{thm:z_B_rho_R}) with the extension lemma and the upper bound in Theorem~\ref{thm:upper_bounds}, we obtain the following result.
\begin{restatable}[Sensitivity to String Reversal]{theorem}{stringreversalsensitivity} \label{thm:construction_string_reversal}
$\Rev_z(n) \in \Theta(\log n)$.
\end{restatable}

Since $c(X) \in O(z(X))$ and $c(X) = c(X^R)$, $\Rev_z(n) \in \Theta(\log n)$ naturally implies $\Ratio_{z,c}(n) \in \Omega(\log n)$, which matches the $O(\log n)$ upper bound.
\begin{restatable}[Worst-Case Ratio of the LZ Factorization to Collage Systems]{corollary}{lzcollageratio} \label{cor:construction_lz_collage_ratio}
$\Ratio_{z,c}(n) \in \Theta(\log n)$.
\end{restatable}
 \section{A Separation between the Lex-Parse and the LZ Factorization}
\label{sse:lex_lz}

In this section, we give a lower bound on the worst-case size ratio of the lex-parse to the LZ factorization by using $B_\id$ with a suitable alphabet order.
Theorem~\ref{thm:z_B_id_rho} already gives $z(B_\id) \in O(m)$, so our goal is to prove $v(B_\id) \in \Omega(m\alpha)$ under this order.
We use a total ordering on $\Sigma_\alpha$ satisfying $\mathtt{b}_i \prec \mathtt{a}_j$ and $\mathtt{b}_i \prec \mathtt{c}_j$ for any $i,j \in \Z_m$.

For an integer $2 \leq q < \alpha$ and $u \in \mathbb{Z}_{2^{\alpha - q}}$,
let $Y_{q,u} = [\max\{1, u 2^q\}, (u+1)2^q - 1]$ be a set of positive integers whose binary representations can be written as $\bin_{\alpha-q}(u) \bin_{q}(w)$, where $w \in \mathbb{Z}_{2^q}$.
We also introduce $X_{q,u} = [u 2^q + 2^{q-1}, u2^q + 2^{q-1} + 2^{q-2} - 1] \subseteq Y_{q,u}$,
the set of positive integers whose binary representations can be written as $\bin_{\alpha-q}(u) \mathtt{10} \bin_{q-2}(v)$, where $v \in \mathbb{Z}_{2^{q - 2}}$.
Let $\ell = 2^{\alpha - q - 1}$.
Fix $x \in X_{q,u}$ such that  $\bin_{\alpha}(x) = \bin_{\alpha-q}(u) \mathtt{10} \bin_{q-2}(v)$ and consider the block $B_x$.
Let $U = \bin_{\alpha-q}(u)$ and $V = \bin_{q-2}(v)$, so $\bin_{\alpha}(x) = U\mathtt{10}V$.
Let $p = \pi_{\alpha}(x)$.
By definition, $B_x[p] = \mathtt{b}_p$.
Figure~\ref{fig:lexparse_ij_m32} illustrates the sets $X_{q,u}$ and $Y_{q,u}$ and the substrings used to compare blocks within $Y_{q,u}$.

\begin{figure}[t]
\centering
\includegraphics[width=\linewidth]{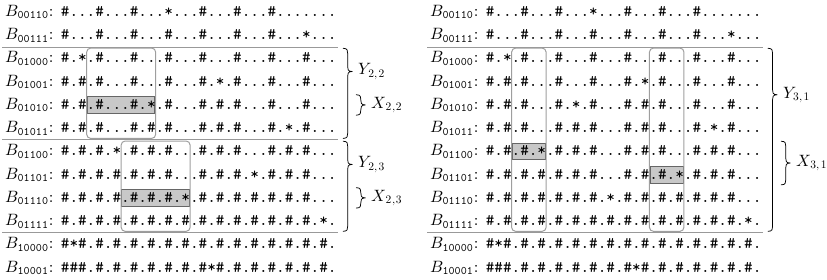}
\caption{
  The sets $X_{q,u}$ and $Y_{q,u}$ for $m = 32$, with $q = 2$ ($\ell = 4$) on the left and $q = 3$ ($\ell = 2$) on the right.
  Both panels show the blocks $B_6, \ldots, B_{17}$ in increasing order, with their indices written in binary.
  As in Figure~\ref{fig:block_orders_m16}, the symbols \texttt{\#}, \texttt{*}, and \texttt{.} represent $\mathtt{a}_i$, $\mathtt{b}_i$, and $\mathtt{c}_i$ at position $i$, respectively.
  For each $x \in X_{q,u}$ and $p = \pi_\alpha(x)$, a rounded rectangle encloses the substrings $B_y[p - 2\ell + 1..p]$ for all $y \in Y_{q,u}$, and the gray box marks $B_x[p - 2\ell + 1..p]$.
  Within each rounded rectangle, the substrings share their first $2\ell - 1$ characters.
  At their last position $p$, only $B_x[p]$ is shown as \texttt{*}, while $B_y[p]$ is shown as \texttt{\#} or \texttt{.} for every other $y$.
}
\label{fig:lexparse_ij_m32}
\end{figure}

The following lemma shows that copy phrases containing positions in $B_x[p - 2\ell + 1..p-1]$ refer to blocks with indices outside $Y_{q,u}$.
\begin{lemma} \label{lem:lex_order_ref}
Let $x \in X_{q,u}$ and $p = \pi_\alpha(x)$.
If a copy phrase in the lex-parse of $B_\id$ contains a position $i \in [p - 2\ell + 1, p-1]$ in $B_x$, then its source occurrence lies in a block $B_y$ with $y \notin Y_{q,u}$.
\end{lemma}
\begin{proof}
We first compare $B_x$ with $B_y$ for $y \in Y_{q,u} \setminus \{x\}$.
For every $s \in Y_{q,u}$, we have $\bin_\alpha(s) = U\bin_q(w)$ for some $w \in \Z_{2^q}$, so $\bin_\alpha(\pi_\alpha(s)) = \bin_q(w)^R U^R$.
The last $\alpha-q = \lg \ell + 1$ bits of $\bin_\alpha(\pi_\alpha(s))$ and $\bin_\alpha(p) = (\bin_\alpha(x))^R$ are both $U^R$, so $\pi_\alpha(s) \equiv p \pmod{2\ell}$ for every $s \in Y_{q,u}$.

By Property~\ref{pro:basic_block_prop} and $\pi_\alpha^{-1} = \pi_\alpha$, the positions where $B_x$ and $B_y$ differ form the set $\{\pi_\alpha(s) \mid \min\{x, y\} \leq s \leq \max\{x, y\}\} \subseteq \{\pi_\alpha(s) \mid s \in Y_{q,u}\}$.
For every $i \in [p - 2\ell + 1, p-1]$, the inequality $0 < |p-i| < 2\ell$ implies $i \not\equiv p \pmod{2\ell}$.
Since $\pi_\alpha(s) \equiv p \pmod{2\ell}$ for every $s \in Y_{q,u}$, we have $i \notin \{\pi_\alpha(s) \mid s \in Y_{q,u}\}$.
Thus, $B_x[p - 2\ell + 1..p-1] = B_y[p - 2\ell + 1..p-1]$.
At position $p = \pi_\alpha(x)$, we have $B_x[p] = \mathtt{b}_p$ and $B_y[p] \in \{\mathtt{a}_p, \mathtt{c}_p\}$ because $\pi_\alpha$ is a bijection and $x \neq y$.
By the alphabet ordering, we have $B_x[p] \prec B_y[p]$.
Thus, $B_x[i..p] \prec B_y[i..p]$ for every $i \in [p - 2\ell + 1, p-1]$.

Let $F$ be a copy phrase in the lex-parse of $B_\id$ containing position $i \in [p - 2\ell + 1, p-1]$ in $B_x$.
By Property~\ref{pro:copy_phrases_in_blocks}, $F$ has the form $B_x[l..r]$ with source occurrence $B_y[l..r]$, where $l \leq i \leq r$ and $x \neq y$.
By the definition of the lex-parse, the suffix of $B_\id$ starting at position $l$ in $B_y$ is lexicographically smaller than the suffix starting at position $l$ in $B_x$.
After we remove their common prefix of length $i-l$, the same inequality holds for the suffixes starting at position $i$ in $B_y$ and $B_x$.
On the other hand, $B_x[i..p] \prec B_y[i..p]$ holds for $y \in Y_{q,u}$, so $y \notin Y_{q,u}$.
\end{proof}

We now prove the size lower bound for the lex-parse by combining Lemma~\ref{lem:block_distance_sum_bound} and Lemma~\ref{lem:lex_order_ref}.

\begin{restatable}{theorem}{lexparsesize} \label{thm:v_B_id}
$v(B_\id) \in \Omega(m\alpha)$.
\end{restatable}
\begin{proof}[Proof (sketch)]
Let $\calE$ be an encoding of the lex-parse of $B_\id$, so $|\calE| = v(B_\id)$.
For each $2 \leq q < \alpha$, each $u \in \Z_{2^{\alpha-q}}$, and each $x \in X_{q,u}$, let $p = \pi_\alpha(x)$, $\ell = 2^{\alpha-q-1}$, and $L_{q,x} = [p-2\ell+1, p-\ell]$.
For fixed $x$, the intervals $L_{q,x}$ are pairwise disjoint over different $q$.

By Lemma~\ref{lem:lex_order_ref}, a copy phrase containing a position $i \in L_{q,x}$ in $B_x$ refers to a block $B_y$ with $y \notin Y_{q,u}$, so $d_\calE(x, i) \geq 2^{q-2}$.
If we sum over all positions in these intervals, the sum of the right-hand sides is $\Theta(m^2)$ for each $q$ and $\Theta(m^2\alpha)$ in total.
This inequality does not hold when the position is represented by a literal phrase.
However, there are at most $|\calE|$ such positions, and the right-hand side at each position is at most $m/8$.
Thus, $\sum_{x=1}^{m-1} \sum_{i=0}^{m-1} d_\calE(x, i) + m|\calE|/8 \in \Omega(m^2\alpha)$.
By the total block-distance lemma, the left-hand side is $O(m|\calE|)$, which gives $|\calE| \in \Omega(m\alpha)$.
\end{proof}

Applying the extension lemma to $v(B_\id) \in \Omega(m\alpha)$ (Theorem~\ref{thm:v_B_id}) and $z(B_\id) \in O(m)$ (Theorem~\ref{thm:z_B_id_rho}), we obtain $\Ratio_{v,z}(n) \in \Omega(\log n)$, matching the $O(\log n)$ upper bound.
\begin{restatable}[Worst-Case Ratio of the Lex-Parse to the LZ Factorization]{theorem}{lexlzratio} \label{thm:construction_lex_lz_ratio}
$\Ratio_{v,z}(n) \in \Theta(\log n)$.
\end{restatable}
 \section{A Separation between Height-Bounded LZ-like Encodings and the LZ Factorization}
\label{sse:heightbounded}

In this section, we give upper and lower bounds on the worst-case size ratio of height-bounded LZ-like encodings to the LZ factorization, with matching bounds for $h \geq \lg n$.

We first state a new upper bound that improves the previous bound $\hat{z}_h(S) \in O(z_S \lg^+(n/z_S))$ for $h \geq c\log n$ for some constant $c$, where $z_S = z(S)$.
The previous bound follows from the size upper bound for run-length grammars~\cite{KociumakaNP23} and their balancing theorem~\cite{NavarroOU22}, together with the conversion to height-bounded LZ-like encodings~\cite{BannaiFHMP24}.
Due to space constraints, we defer the construction and proof for the following theorem to Appendix~\ref{sec:height_bounded_upper_bounds}.
\begin{restatable}{theorem}{generalstringheightupperbound} \label{thm:general_string_height_upper_bound}
Let $S$ be a nonempty string of length $n$, and let $z_S = z(S)$.
For every positive integer $h \geq \lg(n/z_S)$, there exists an LZ-like encoding $\calE$ of $S$ with $\height(\calE) \leq h$ and $|\calE| \in O\left(z_S\left(1 + \frac{\lg^+(n/z_S)}{\lg(1 + h/\lg^+(n/z_S))}\right)\right)$.
\end{restatable}

To prove the lower bound, we combine the total block-distance lemma with a cost analysis of \emph{height-bounded increasing trees}.
We first define increasing trees and give a lower bound on their minimum cost under a height constraint.
We then use this cost bound to analyze the size of height-bounded LZ-like encodings of $B_\id$.

\subsection{Height-Bounded Increasing Trees}

An \emph{increasing tree}~\cite{BergeronFS92} on $r$ vertices is a rooted tree $\calT$ with vertex set $\Z_r$ and root $0$ such that every vertex $i \in \Z_r \setminus \{0\}$ has a parent $p_{\calT}(i) < i$.
The \emph{depth} of each vertex $i \in \Z_r$ is the number of edges on the path from $i$ to the root and is denoted by $\depth_{\calT}(i)$.
The \emph{height} of $\calT$ is $\height(\calT) = \max_{i \in \Z_r} \depth_{\calT}(i)$.
The edge between $i$ and $p_{\calT}(i)$ has length $i - p_{\calT}(i)$.
The \emph{cost} of $\calT$ is defined as $\cost(\calT) = \sum_{i = 1}^{r - 1} (i - p_{\calT}(i))$.
For integers $r \geq 1$ and $h \geq 0$, let $C(r, h)$ be the minimum cost of an increasing tree on $r$ vertices with height at most $h$.
The only increasing tree of height $0$ consists of a single vertex, so $C(1, 0) = 0$ and $C(r, 0) = \infty$ for every $r \geq 2$.
For all $r \geq 1$ and $h \geq 0$, $C(r, h) \leq C(r + 1, h)$ and $C(r, h) \geq C(r, h + 1)$ hold.

For positive integers $r$ and $h$, we define $Q(r, h) = \min\{q \in \Z_{\geq 0} \mid \binom{h+q}{h} \geq r\}$.
The following two theorems give asymptotically tight bounds on $Q(r, h)$ and $C(r, h)$, respectively.
Both proofs are given in Appendix~\ref{sec:increasing_tree_cost_bounds}.
\begin{restatable}{theorem}{orderEvalQ} \label{thm:order_eval_Q}
For integers $r \geq 2$ and $h \geq 1$,
\[
Q(r, h) \in
\begin{cases}
\Theta\left(h r^{1/h}\right) & \text{if } h \leq \ln r, \\
\Theta\left(1 + \dfrac{\ln r}{\ln(1 + h/\ln r)}\right) & \text{if } h > \ln r.
\end{cases}
\]
\end{restatable}

\begin{restatable}{theorem}{increasingtreeexactcost} \label{thm:increasing_tree_exact_cost}
For integers $r \geq 2$ and $h \geq 1$, $\frac{1}{2} r Q(r, h) \leq C(r, h) < r Q(r, h)$.
\end{restatable}

\subsection{Lower Bound on Height-Bounded LZ-like Encodings}

Let $\calE$ be an arbitrary LZ-like encoding of $B_\id$.

For each $i \in \Z_m$, we define an increasing tree $\calT_{\calE,i}$ with vertex set $\Z_m$ and root $0$.
For each $1 \leq x < m$, we set $p_{\calT_{\calE,i}}(x) = 0$ if $B_x[i]$ is represented by a literal phrase and $p_{\calT_{\calE,i}}(x) = y$ if its transition refers to $B_y[i]$.
By Property~\ref{pro:copy_phrases_in_blocks}, each copy phrase refers to an occurrence at the same offset in a different block.
Since $B_y$ occurs before $B_x$, the label of the parent of every vertex other than the root is smaller than the label of the vertex, so $\calT_{\calE,i}$ is a valid increasing tree.
Figure~\ref{fig:height_bounded_trees_m8} illustrates these increasing trees for an LZ-like encoding of $B_\id$.

\begin{figure}[t]
\centering
\includegraphics[width=\linewidth]{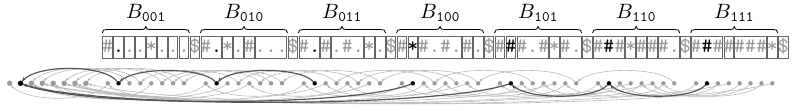}
\caption{
  An example of an LZ-like factorization of $B_\id$ for $\alpha = 3$ and $m = 8$, together with the increasing trees $\calT_{\calE,i}$ for the corresponding encoding $\calE$.
  Each box corresponds to a phrase, and block indices are written in binary.
  As in Figure~\ref{fig:block_orders_m16}, the symbols \texttt{\#}, \texttt{*}, and \texttt{.} represent $\mathtt{a}_i$, $\mathtt{b}_i$, and $\mathtt{c}_i$ at offset $i$, respectively.
  The subscript of each delimiter is omitted.
  The roots of the increasing trees are drawn to the left of the string in increasing order of offset.
  The characters highlighted in black mark the second position in each block, and the darker edges show the corresponding increasing tree $\calT_{\calE,1}$.
  This increasing tree has height $\height(\calT_{\calE,1}) = 3$.
}
\label{fig:height_bounded_trees_m8}
\end{figure}

We bound the total cost of these trees by counting the edge lengths for each phrase.
\begin{lemma} \label{lem:construction_increasing_tree_cost}
For every LZ-like encoding $\calE$ of $B_\id$,
\[
\sum_{i=0}^{m-1} \cost(\calT_{\calE,i}) \leq 4m|\calE|.
\]
\end{lemma}
\begin{proof}
For each $1 \leq x < m$, vertex $x$ in $\calT_{\calE,i}$ corresponds to position $i$ in $B_x$.
Consider a copy phrase $B_x[l..r]$ with source occurrence $B_y[l..r]$.
For $l \leq i \leq r$, vertex $x$ in $\calT_{\calE,i}$ has parent $y$.
By Lemma~\ref{lem:construction_block_lcs_bound}, the total length of these edges is $\sum_{i=l}^{r} d_\calE(x, i) = (x-y)(r-l+1) \leq 4m$.
A literal phrase in a block $B_x$ corresponds to a vertex whose edge to the root has length $x \leq m-1 < 4m$.
Thus, the total edge length associated with each phrase is at most $4m$.
Since there are $|\calE|$ phrases, we have $\sum_{i=0}^{m-1} \cost(\calT_{\calE,i}) \leq 4m|\calE|$.
\end{proof}

The following theorem gives the size lower bound in terms of $C$.
\begin{theorem} \label{thm:zh_B_id}
For every integer $h \geq 1$, $\hat{z}_h(B_\id) \in \Omega(C(m, h+1)) = \Omega(mQ(m, h+1))$.
\end{theorem}
\begin{proof}
Let $\calE$ be an LZ-like encoding of $B_\id$ with $|\calE| = \hat{z}_h(B_\id)$ and $\height(\calE) \leq h$.
We first bound the height of each tree $\calT_{\calE,i}$.
Consider the path from a vertex $x \in \Z_m \setminus \{0\}$ to the root $0$ in $\calT_{\calE,i}$.
Each edge on this path except the edge to the root represents one transition under $\tau_\calE$.
The starting vertex of the final edge corresponds to a literal phrase.
Thus, $\depth_{\calT_{\calE,i}}(x) = \depth_\calE((x-1)(m+1)+i) + 1 \leq h+1$ for every $1 \leq x < m$.
This proves $\height(\calT_{\calE,i}) \leq h+1$.

Each of the $m$ trees $\calT_{\calE,i}$ has $m$ vertices and height at most $h+1$, so its cost is at least $C(m, h+1)$.
Combining these lower bounds with Lemma~\ref{lem:construction_increasing_tree_cost}, we obtain
$m C(m, h+1)
\leq \sum_{i=0}^{m-1} \cost(\calT_{\calE,i})
\leq 4m|\calE|
= 4m\hat{z}_h(B_\id)$.
Dividing by $4m$ proves $\hat{z}_h(B_\id) \in \Omega(C(m, h+1)) = \Omega(mQ(m, h+1))$, where the equality follows from Theorem~\ref{thm:increasing_tree_exact_cost}.
\end{proof}

By applying the extension lemma and evaluating $Q(r, h)$ under the respective conditions on $h$, we obtain the following results on $\Ratio_{\hat{z}_h,z}(n)$.
\begin{restatable}[Worst-Case Ratio of Height-Bounded LZ-like Encodings to the LZ Factorization]{theorem}{heightboundedlzratio} \label{thm:construction_height_bounded_ratio}
For every sufficiently large integer $n$, let $\alpha \geq 3$ be the largest integer such that $2^{2\alpha} < n$.
Let $m = 2^\alpha$.
For every integer $h \geq 1$,
\[
\begin{aligned}
\Ratio_{\hat{z}_h,z}(n)
&\in \Omega\left(Q(\lfloor\sqrt{n}\rfloor, h+1)\right) \\
&= \begin{cases}
\Omega\left((h+1)n^{1/(2(h+1))}\right) & \text{if } h+1 \leq \ln m, \\
\Omega\left(1 + \dfrac{\ln m}{\ln(1 + (h+1)/\ln m)}\right) & \text{if } h+1 > \ln m.
\end{cases}
\end{aligned}
\]
In particular, for any function $h \colon \mathbb{Z}_{\geq 1} \to \mathbb{Z}_{\geq 1}$, the following statements hold.
\begin{enumerate}
\renewcommand{\theenumi}{\roman{enumi}}
\renewcommand{\labelenumi}{(\theenumi)}
  \item \label{enm:construction_height_log} If $h(n) \in O(\log n)$, then $\Ratio_{\hat{z}_{h(n)},z}(n) \in \Omega(\log n)$.
  \item \label{enm:construction_height_polylog} If $h(n) \in O(\polylog n)$, then $\Ratio_{\hat{z}_{h(n)},z}(n) \in \Omega\left(\frac{\log n}{\log\log n}\right)$.
  \item \label{enm:construction_height_constant} If $h(n) \leq H$ for every $n$ and some fixed integer $H \geq 1$, then $\Ratio_{\hat{z}_{h(n)},z}(n) \in \Omega\left(n^{1/(2(H+1))}\right)$.
  \item \label{enm:construction_height_constant_ratio} If $\Ratio_{\hat{z}_{h(n)},z}(n) \in O(1)$, then $h(n) \in n^{\Omega(1)}$.
\end{enumerate}
Moreover, for every integer $h \geq \lg n$,
\[
\Ratio_{\hat{z}_h,z}(n)
\in \Theta\left(1 + \frac{\lg n}{\lg(1 + h/\lg n)}\right).
\]
\end{restatable}
\begin{proof}[Proof (sketch)]
The lower bounds follow from Theorems~\ref{thm:z_B_id_rho}, \ref{thm:order_eval_Q}, \ref{thm:increasing_tree_exact_cost}, and~\ref{thm:zh_B_id} and Lemma~\ref{lem:extension_lemma}.
Statements (i)--(iv) follow from the bounds on $Q(r, h)$ under the respective conditions on $h$.
Theorem~\ref{thm:general_string_height_upper_bound} gives the matching upper bound for $h \geq \lg n$.
\end{proof}
 \section*{AI Usage Disclosure}
The authors used OpenAI's GPT-5.6 Sol and GPT-6 Astra for assistance with research and manuscript preparation.
In particular, the initial string construction and the proofs in Section~\ref{sse:pref_deletion}, Section~\ref{sse:lex_lz}, and Section~\ref{sse:heightbounded} were developed with assistance from these models.
The proof of the upper bound in Appendix~\ref{sec:height_bounded_upper_bounds} was also developed with their assistance.
These proofs were then simplified through discussions with the models.
The models were also used for drafting and correcting the text, improving its presentation, and preparing and editing the figure files.
The authors verified all claims and proofs and take full responsibility for the contents of this paper.
 
\clearpage
\bibliographystyle{plain}
\bibliography{ref}

\clearpage
\appendix
\section{Omitted Proofs} \label{sec:omitted_sensitivity_ratio_proofs}

\substringcomplexitybounds*
\begin{proof}
The lower and upper bounds for $f \in \{z, c, v\}$ are shown in \cite{Navarro21a,MigitaUI26}.
For $z_{\mathrm{rel}}$, we have
\[
z_{\mathrm{rel}}(T)
= \min_{S \in \Sigma^*} \{z(S) + z(T \mid S)\}
\geq \min_{S \in \Sigma^*} z(ST)
\geq \min_{S \in \Sigma^*} \delta(ST)
= \delta(T),
\]
where the last equality follows from the monotonicity of $\delta$ and the choice $S = \varepsilon$.
Choosing $T$ itself as the reference gives $z_{\mathrm{rel}}(T) \leq z(T) + z(T \mid T) \leq z(T) + 1$, so the upper bound for $z$ also holds for $z_{\mathrm{rel}}$.

For any $f, g \in \{z, c, v, z_{\mathrm{rel}}\}$ and any string $T$ of length $n \geq 2$, these bounds give
\[
\frac{f(T)}{g(T)}
\in O\left(\log\left(1 + \frac{n}{\delta(T)}\right)\right)
\subseteq O(\log n),
\]
since $\delta(T) \geq 1$.
Hence, $\Ratio_{f,g}(n) \in O(\log n)$.
\end{proof}

\sensitivityupperbounds*
\begin{proof}
Fix an integer $n \geq 2$.
The bounds on $z$ in terms of $\delta$ imply that there is a constant $\alpha > 0$ such that
$z(X) \leq \alpha \delta(X) \log n$ for every nonempty string $X$ of length at most $n$.

For prefix deletion, $T$ is a substring of $ST$, so the monotonicity of $\delta$ gives $\delta(T) \leq \delta(ST)$.
Therefore,
\[
z(T) \leq \alpha \delta(T) \log n
\leq \alpha \delta(ST) \log n
\leq \alpha z(ST) \log n.
\]

For string reversal, the symmetry of $\delta$ gives
\[
z(T^R)
\leq \alpha \delta(T^R) \log n
= \alpha \delta(T) \log n
\leq \alpha z(T) \log n.
\]

For cyclic rotation, both $S$ and $T$ are substrings of $ST$.
The monotonicity of $\delta$ gives $\delta(S) \leq \delta(ST)$ and $\delta(T) \leq \delta(ST)$.
Since $\delta(ST) \geq 1$, the concatenation inequality above gives
\[
\delta(TS)
\leq \delta(S) + \delta(T) + 1
\leq 2\delta(ST) + 1
\leq 3\delta(ST).
\]
As in the prefix-deletion case, the bounds relating $z$ and $\delta$ give $z(TS) \leq 3\alpha z(ST) \log n$.
All sensitivities are therefore in $O(\log n)$.
\end{proof}

\blockdistancelemma*
\begin{proof}
Assume $x \leq y$ without loss of generality.
The claim is immediate if $\ell = 0$ or $d = 0$, so suppose that $\ell > 0$ and $d = y - x > 0$.
By Property~\ref{pro:copy_phrases_in_blocks}, the two occurrences of any nonempty common substring of these blocks start at the same offset.
Let $B_x[l..r] = B_y[l..r]$ be a longest common substring between the two blocks, where $\ell = r-l+1$.

If $d \leq 3$, then $(d+1)\ell \leq 4m$ follows from $\ell \leq m$.
For the remaining case $d \geq 4$, suppose for contradiction that $\ell > 4m/(d+1)$.
We show that $[l, r]$ contains a position $p$ satisfying $x \leq \pi_\alpha(p) \leq y$.
Let $q = \lfloor\lg\ell\rfloor-1$ and $a = \lceil l/2^q\rceil$.
By the definitions of $q$ and $a$, we have $2^{q+1} \leq \ell < 2^{q+2}$ and $l \leq a2^q < l+2^q$.
These inequalities imply that the interval $I = [a2^q, (a+1)2^q-1] = \{ a2^q + b \mid b \in \Z_{2^q} \}$ is contained in $[l, r]$.
For every $b \in \Z_{2^q}$, the binary representation $\bin_\alpha(a2^q+b)$ is the concatenation of $\bin_{\alpha-q}(a)$ and $\bin_q(b)$.
By reversing the bits of this representation, we obtain
\begin{align*}
\pi_\alpha(a2^q+b)
&= \int{\left(\left(\bin_{\alpha-q}(a)\bin_q(b)\right)^R\right)} \\
&= \int{\left(\bin_q(b)^R\bin_{\alpha-q}(a)^R\right)} \\
&= 2^{\alpha-q}\int{\left(\bin_q(b)^R\right)} + \int{\left(\bin_{\alpha-q}(a)^R\right)} \\
&= 2^{\alpha-q}\pi_q(b) + \pi_{\alpha-q}(a).
\end{align*}
Since $\pi_q$ is a permutation of $\Z_{2^q}$, we have
\begin{align*}
\{\pi_\alpha(i) \mid i \in I\}
&= \{\pi_\alpha(a2^q+b) \mid b \in \Z_{2^q}\} \\
&= \{2^{\alpha-q}\pi_q(b) + \pi_{\alpha-q}(a) \mid b \in \Z_{2^q}\} \\
&= \{2^{\alpha-q}b + \pi_{\alpha-q}(a) \mid b \in \Z_{2^q}\}.
\end{align*}
Thus, every interval of $2^{\alpha-q}$ consecutive integers in $[0, m-1]$ contains an element of this set.

To show that $B_x$ and $B_y$ differ at some position in $I \subseteq [l, r]$, we prove that $[x, y]$ contains an element of $\{\pi_\alpha(i) \mid i \in I\}$.
The inequalities $\ell < 2^{q+2}$ and $\ell > 4m/(d+1)$ imply $2^{\alpha-q} = m/2^q < 4m/\ell < d+1 = y - x + 1$.
Thus, the interval $[x, y]$ contains at least $2^{\alpha-q}$ consecutive integers, so there exists a position $p \in I \subseteq [l, r]$ such that $\pi_\alpha(p) \in [x, y]$.
For this position $p \in I \subseteq [l, r]$, Property~\ref{pro:basic_block_prop} implies $B_x[p] \neq B_y[p]$, which contradicts $B_x[l..r] = B_y[l..r]$.
\end{proof}

\extensionlemma*
\begin{proof}
Let $\ell = n'-n$ and $T' = \mathtt{\#}^{\ell}T$, where $\mathtt{\#}$ is a new character greater than every character in $T$ under $\prec$.
For any string $X$ not containing $\mathtt{\#}$, the LZ factorization of $\mathtt{\#}^{\ell}X$ consists of the two phrases $\mathtt{\#}$ and $\mathtt{\#}^{\ell-1}$ followed by the phrases of the LZ factorization of $X$.
The LZ factorization of $X\mathtt{\#}^{\ell}$ consists of the phrases of the LZ factorization of $X$ followed by these two phrases.
Thus, $z(T') = z(T) + 2$ and $z\!\left((T')^R\right) = z(T^R) + 2$.

For $0 \leq i < j < \ell$, the choice of $\mathtt{\#}$ implies $T'[j..n'-1] \prec T'[i..n'-1]$.
The longest common prefix of $T'$ and $T'[1..n'-1]$ has length $\ell-1$, and no suffix starting after position $1$ has a longer common prefix with $T'$.
Thus, the lex-parse of $T'$ begins with a copy phrase $\mathtt{\#}^{\ell-1}$ with source position $1$ followed by a literal phrase $\mathtt{\#}$.
Every suffix of $T'$ starting in $T$ is identical to the corresponding suffix of $T$.
Since none of these suffixes starts with $\mathtt{\#}$, the remaining phrases are exactly those of the lex-parse of $T$, so $v(T') = v(T) + 2$.

Fix an integer $h \geq 1$.
The prefix $\mathtt{\#}^{\ell}$ has an LZ-like encoding of height $1$ consisting of a literal phrase $\mathtt{\#}$ followed by a copy phrase $\mathtt{\#}^{\ell-1}$ with source position $0$.
Combining these two phrases with a smallest LZ-like encoding of $T$ of height at most $h$ gives $\hat{z}_h(T') \leq \hat{z}_h(T) + 2$.
Conversely, let $\calE'$ be an LZ-like encoding of $T'$ of height at most $h$.
If a copy phrase starts at a position $p < \ell$ and contains position $\ell$, its source position $s < p$ satisfies $T'[s+\ell-p] = \mathtt{\#}$, whereas $T'[\ell] \neq \mathtt{\#}$.
Thus, no phrase crosses the boundary between $\mathtt{\#}^{\ell}$ and $T$.
Every source occurrence of a copy phrase in $T$ starts within $T$ because $\mathtt{\#}$ does not occur in $T$.
The phrases covering $T$ therefore define an LZ-like encoding of $T$ with the same depths after shifting their starting positions and the source positions of their copy phrases by $-\ell$.
Since the prefix $\mathtt{\#}^{\ell}$ requires at least two phrases, we have $|\calE'| \geq \hat{z}_h(T) + 2$.
This proves $\hat{z}_h(T') = \hat{z}_h(T) + 2$.

Finally, let $S$ be a string satisfying $z_{\mathrm{rel}}(T) = z(S) + z(T \mid S)$ and set $S' = \mathtt{\#}^{\ell}S$.
The two phrases representing $\mathtt{\#}^{\ell}$ followed by the LZ factorization of $S$ give an LZ-like factorization of $S'$, so $z(S') \leq z(S) + 2$.
The prefix $\mathtt{\#}^{\ell}$ of $T'$ occurs in $S'$, and every copy phrase of the RLZ factorization of $T$ with respect to $S$ also has an occurrence in $S'$.
Thus, $T'$ can be represented by one copy phrase for the prefix followed by the RLZ phrases of $T$, giving $z(T' \mid S') \leq z(T \mid S) + 1$.
Hence, $z_{\mathrm{rel}}(T') \leq z(S') + z(T' \mid S') \leq z_{\mathrm{rel}}(T) + 3$.

\end{proof}

\prefixdeletionsensitivity*
\begin{proof}
We first prove the lower bound.
For a sufficiently large integer $n$, let $\alpha$ be the largest integer such that $2^{2\alpha+1} \leq n$, and let $m = 2^\alpha$.
Then $2m^2 \leq n < 8m^2$, so $\alpha = \lg m \in \Theta(\log n)$.
By Theorem~\ref{thm:z_B_id_rho} and Theorem~\ref{thm:z_B_rho}, there exist constants $c_1, c_2 > 0$ such that $z(B_\id B_\rho) \leq c_1 m$ and $z(B_\rho) \geq c_2 m\alpha$.

Since $|B_\id B_\rho| = 2(m^2-1) \leq n-2$, applying the extension lemma (Lemma~\ref{lem:extension_lemma}) to $B_\id B_\rho$ yields a string $ST$ of length $n$ whose suffix $T$ is $B_\rho$.
By the lemma, $z(ST) = z(B_\id B_\rho) + 2$.

Therefore,
\[
\DPre_z(n)
\geq \frac{z(T)}{z(ST)}
\geq \frac{c_2m\alpha}{c_1m + 2}
\geq \frac{c_2}{c_1 + 2}\alpha.
\]
Since $\alpha \in \Theta(\log n)$, we have $\DPre_z(n) \in \Omega(\log n)$.
Combining this lower bound with the upper bound in Theorem~\ref{thm:upper_bounds}, we obtain $\DPre_z(n) \in \Theta(\log n)$.
\end{proof}

\reversedrhosize*
\begin{proof}
We construct a variant of an LZ-like factorization of $B_\rho$ with $O(m)$ phrases, where each copy phrase refers to a later occurrence.

Consider a block $B_x$ in $B_\rho$ with $1 \leq x < m-1$.
If $x$ is even, then $\lsb(x) > 0$ and $\lsb(x+1) = 0$.
Thus, $B_{x+1}$ occurs after $B_x$ in $B_\rho$ by the definition of $\rho$.
By Property~\ref{pro:basic_block_prop}, the blocks $B_x$ and $B_{x+1}$ have different characters only at positions $\pi(x)$ and $\pi(x+1)$.
We can represent $B_x$ with at most five phrases by using literal phrases at these two positions and referring to the corresponding substrings of $B_{x+1}$ for the remaining at most three substrings.

If $x$ is odd, then $\lsb(x) = \lsb(x+2) = 0$.
In this case, $B_{x+2}$ occurs immediately after $B_x$ in $B_\rho$ by the definition of $\rho$.
Since these blocks have different characters only at positions $\pi(x)$, $\pi(x+1)$, and $\pi(x+2)$, we represent $B_x$ by at most seven phrases in the same way.

The last block $B_{\rho(m-1)} = B_{m-1}$ and the $m-1$ delimiters can be represented by $2m-1$ literal phrases in total.
There are $m/2-1$ even indices and $m/2-1$ odd indices in $[1, m-2]$, so the total number of phrases is at most $5(m/2-1) + 7(m/2-1) + (2m-1) = 8m-13$.
By reversing the phrase order and assigning each reference to the corresponding interval, we obtain an LZ-like factorization of $B_\rho^R$ and hence $z(B_\rho^R) \in O(m)$.
\end{proof}

\stringreversalsensitivity*
\begin{proof}
The upper bound follows from Theorem~\ref{thm:upper_bounds}.

To prove the lower bound, for any sufficiently large $n$, we construct a string $T$ of length $n$ with $z(T^R) / z(T) \in \Omega(\log n)$.
Let $\alpha \geq 3$ be the largest integer such that $2^{2\alpha} < n$, and let $m = 2^\alpha$.
The choice of $\alpha$ implies $m^2 < n \leq 4m^2$, so $\alpha = \lg m \in \Theta(\log n)$.
By Theorem~\ref{thm:z_B_rho_R} and Theorem~\ref{thm:z_B_rho}, there exist constants $c_1, c_2 > 0$ such that $z(B_\rho^R) \leq c_1 m$ and $z(B_\rho) \geq c_2 m\alpha$.

Since $|B_\rho^R| = m^2-1 \leq n-2$, applying the extension lemma (Lemma~\ref{lem:extension_lemma}) to $B_\rho^R$ yields a string $T$ of length $n$.
By the lemma, $z(T) = z(B_\rho^R) + 2$ and $z(T^R) = z(B_\rho) + 2$.

Therefore,
\[
\Rev_z(n)
\geq \frac{z(T^R)}{z(T)}
\geq \frac{c_2 m\alpha + 2}{c_1 m + 2}
\geq \frac{c_2}{c_1 + 2}\alpha.
\]
Since $\alpha \in \Theta(\log n)$, this proves $\Rev_z(n) \in \Omega(\log n)$.
\end{proof}

\lexparsesize*
\begin{proof}
For each $2 \leq q < \alpha$ and each $x \in X_{q,u}$ for some $u \in \Z_{2^{\alpha-q}}$, let $p = \pi_\alpha(x)$, $\ell = 2^{\alpha-q-1}$, and $L_{q,x} = [p-2\ell+1, p-\ell]$.
Since $x \in X_{q,u} \subseteq [u2^q, (u+1)2^q-1]$, the index $u = \lfloor x/2^q \rfloor$ is unique for fixed $q$ and $x$.
By the definition of $X_{q,u}$, there are binary strings $U$ and $V$ with $|U| = \alpha-q$ such that $\bin_\alpha(x) = U\mathtt{10}V$ and $\bin_\alpha(p) = V^R\mathtt{01}U^R$.
Thus, $p \geq 2^{\alpha-q} = 2\ell$ and $p-2\ell+1 \geq 0$, which ensures $L_{q,x} \subseteq [0, m-1]$.
For fixed $x$, each position $i \in L_{q,x}$ satisfies $2^{\alpha-q-1} \leq p-i < 2^{\alpha-q}$.
These ranges of $p-i$ are disjoint for different $q$, so the intervals $L_{q,x}$ are pairwise disjoint.

Let $\calE$ be an encoding of the lex-parse of $B_\id$, so $|\calE| = v(B_\id)$.
We give a lower bound on the total sum of $d_\calE(x, i)$ and then obtain a lower bound on $|\calE|$ using Lemma~\ref{lem:block_distance_sum_bound}.
We define a function $f$ on $(\Z_m \setminus \{0\}) \times \Z_m$ by $f(x, i) = 2^{q-2}$ for $i \in L_{q,x}$ and $f(x, i) = 0$ otherwise.

By Lemma~\ref{lem:lex_order_ref}, a copy phrase containing a position $i \in L_{q,x}$ in $B_x$ refers to an occurrence in a block $B_y$ with $y \notin Y_{q,u}$.
The definitions of $X_{q,u}$ and $Y_{q,u}$ imply $d_\calE(x, i) = |x - y| \geq 2^{q-2} = f(x, i)$.
Thus, $f(x, i)$ gives a lower bound for $d_\calE(x, i)$ at every position in a copy phrase.
For a position $i$ in $B_x$ represented by a literal phrase, $d_\calE(x, i) = 0$.
For such a position, the definition of $f$ gives $f(x, i) - d_\calE(x, i) = f(x, i) \leq 2^{\alpha-3} = m/8$.
Since each literal phrase has length $1$, at most $|\calE|$ positions in the blocks are represented by literal phrases.
Thus, $\sum_{x=1}^{m-1} \sum_{i=0}^{m-1} d_\calE(x, i) \geq \sum_{x=1}^{m-1} \sum_{i=0}^{m-1} f(x, i) - \frac{m}{8}|\calE|$.

For each $q$, there are $2^{\alpha-q}$ possible values of $u$ and $2^{q-2}$ indices $x$ in each $X_{q,u}$.
Each interval $L_{q,x}$ has length $2^{\alpha-q-1}$.
Combining these bounds with the disjointness of the intervals in each block and Lemma~\ref{lem:block_distance_sum_bound}, we obtain
\begin{align*}
4m|\calE|
&\geq \sum_{x=1}^{m-1} \sum_{i=0}^{m-1} d_\calE(x, i) \\
&\geq \sum_{x=1}^{m-1} \sum_{i=0}^{m-1} f(x, i) - \frac{m}{8}|\calE| \\
&= \sum_{q=2}^{\alpha-1} \sum_{u=0}^{2^{\alpha-q}-1} \sum_{x \in X_{q,u}} 2^{q-2}|L_{q,x}| - \frac{m}{8}|\calE| \\
&= \sum_{q=2}^{\alpha-1} 2^{\alpha-q} \cdot 2^{q-2} \cdot 2^{q-2} \cdot 2^{\alpha-q-1} - \frac{m}{8}|\calE| \\
&= \frac{m^2(\alpha-2)}{32} - \frac{m}{8}|\calE|.
\end{align*}
Thus, $\frac{33}{8}|\calE| \geq \frac{m(\alpha-2)}{32}$.
Since $|\calE| = v(B_\id)$, we obtain $v(B_\id) \in \Omega(m\alpha)$.
\end{proof}

\lexlzratio*
\begin{proof}
The upper bound follows from the inequalities $\delta(T) \in O(z(T))$ and $v(T) \in O(\delta(T) \log n)$ for strings $T$ of length $n$.

To prove the lower bound, for any sufficiently large $n$, we construct a string $T$ of length $n$ with $v(T) / z(T) \in \Omega(\log n)$.
Let $\alpha \geq 3$ be the largest integer such that $2^{2\alpha} < n$, and let $m = 2^\alpha$.
The choice of $\alpha$ implies $m^2 < n \leq 4m^2$, so $\alpha = \lg m \in \Theta(\log n)$.
By Theorem~\ref{thm:z_B_id_rho} and Theorem~\ref{thm:v_B_id}, there exist constants $c_1, c_2 > 0$ such that $z(B_\id) \leq c_1 m$ and $v(B_\id) \geq c_2 m\alpha$.

Since $|B_\id| = m^2-1 \leq n-2$, applying the extension lemma (Lemma~\ref{lem:extension_lemma}) to $B_\id$ yields a string $T$ of length $n$.
By the lemma, $z(T) = z(B_\id) + 2$ and $v(T) = v(B_\id) + 2$.

Therefore,
\[
\Ratio_{v,z}(n)
\geq \frac{v(T)}{z(T)}
\geq \frac{c_2 m\alpha + 2}{c_1 m + 2}
\geq \frac{c_2}{c_1 + 2}\alpha.
\]
Since $\alpha \in \Theta(\log n)$, this proves $\Ratio_{v,z}(n) \in \Omega(\log n)$.
\end{proof}

\heightboundedlzratio*
\begin{proof}
We first construct a string of length $n$ for the lower bound.
The choice of $\alpha$ implies $m^2 < n \leq 4m^2$.
Let $R = \ln m \in \Theta(\log n)$.
By Theorem~\ref{thm:z_B_id_rho} and Theorem~\ref{thm:zh_B_id}, there exist constants $c_1, c_2 > 0$ such that $z(B_\id) \leq c_1 m$ and $\hat{z}_h(B_\id) \geq c_2 m Q(m, h+1)$ for every integer $h \geq 1$.

Since $|B_\id| = m^2-1 \leq n-2$, applying the extension lemma (Lemma~\ref{lem:extension_lemma}) to $B_\id$ yields a string $T$ of length $n$.
By the lemma, $z(T) = z(B_\id) + 2$ and $\hat{z}_h(T) = \hat{z}_h(B_\id) + 2$ for every integer $h \geq 1$.
Therefore,
\[
\Ratio_{\hat{z}_h,z}(n)
\geq \frac{\hat{z}_h(T)}{z(T)}
\geq \frac{c_2 m Q(m, h+1) + 2}{c_1 m + 2}
\geq \frac{c_2}{c_1 + 2}Q(m, h+1).
\]
Let $r = \lfloor\sqrt{n}\rfloor$.
Since $m \leq r \leq \sqrt{n} \leq 2m$, Theorems~\ref{thm:increasing_tree_exact_cost} and~\ref{thm:order_eval_Q} give
\[
\frac{C(r, h+1)}{\sqrt{n}}
\in \Theta\left(Q(r, h+1)\right)
= \Theta\left(Q(m, h+1)\right).
\]
Applying the bounds on $Q(m, h+1)$ in Theorem~\ref{thm:order_eval_Q} gives the two cases in the statement.
For the first case, the inequalities $n/4 \leq m^2 < n$ imply $m^{1/(h+1)} \in \Theta\left(n^{1/(2(h+1))}\right)$ for $h \geq 1$.

We next prove statements~(\ref{enm:construction_height_log})--(\ref{enm:construction_height_constant_ratio}) and abbreviate $h(n)$ as $h$.
When $h+1 \leq R$, the inequality $(h+1)e^{R/(h+1)} \geq R$ gives $\Ratio_{\hat{z}_h,z}(n) \in \Omega(R) = \Omega(\log n)$.
For statement~(\ref{enm:construction_height_log}), the assumption $h \in O(\log n) = O(R)$ also implies $\ln(1 + (h+1)/R) \in O(1)$.
Thus, the second case gives the same lower bound $\Omega(\log n)$ when $h+1 > R$.

For statement~(\ref{enm:construction_height_polylog}), suppose that $h \in O(\polylog n)$.
The case $h+1 \leq R$ already gives $\Ratio_{\hat{z}_h,z}(n) \in \Omega(\log n)$.
When $h+1 > R$, the assumptions on $h$ and $R$ imply $\ln(1 + (h+1)/R) \in O(\log\log n)$.
The second case therefore gives $\Ratio_{\hat{z}_h,z}(n) \in \Omega(\log n/\log\log n)$.

For statement~(\ref{enm:construction_height_constant}), let $H \geq 1$ be fixed and suppose that $h \leq H$ for every $n$.
For each string $T$, the minimum size $\hat{z}_h(T)$ is nonincreasing in $h$.
Hence, $\Ratio_{\hat{z}_h,z}(n) \geq \Ratio_{\hat{z}_H,z}(n)$.
Since $H+1 \leq R$ for all sufficiently large $n$, the first case gives
\[
\Ratio_{\hat{z}_h,z}(n)
\geq \Ratio_{\hat{z}_H,z}(n)
\in \Omega\left((H+1)n^{1/(2(H+1))}\right)
= \Omega\left(n^{1/(2(H+1))}\right).
\]

For statement~(\ref{enm:construction_height_constant_ratio}), suppose that $\Ratio_{\hat{z}_h,z}(n) \in O(1)$.
The lower bound $\Omega(\log n)$ for $h+1 \leq R$ implies that $h+1 > R$ for all sufficiently large $n$.
The second case then implies $\ln(1 + (h+1)/R) \geq cR$ for some constant $c > 0$ and all sufficiently large $n$.
Thus, $h \geq R(e^{cR}-1)-1$.
Since $R \in \Theta(\log n)$, this proves $h \in n^{\Omega(1)}$.

Finally, suppose that $h \geq \lg n$.
The inequalities $h+1 > R$ and $h \leq h+1 \leq 2h$, together with $R \in \Theta(\log n)$, give
\[
\Ratio_{\hat{z}_h,z}(n)
\in \Omega\left(1 + \frac{R}{\ln(1 + (h+1)/R)}\right)
= \Omega\left(1 + \frac{\lg n}{\lg(1 + h/\lg n)}\right).
\]

For the upper bound, let $T$ be any string of length $n$ and let $L = \lg^+(n/z(T))$.
Since $1 \leq L \leq \lg n$, Theorem~\ref{thm:general_string_height_upper_bound} gives
\[
\frac{\hat{z}_h(T)}{z(T)}
\in O\left(1 + \frac{L}{\lg(1 + h/L)}\right)
\subseteq O\left(1 + \frac{\lg n}{\lg(1 + h/\lg n)}\right).
\]
Taking the maximum over all strings $T$ of length $n$ proves the upper bound.
\end{proof}

\cyclicrotationsensitivity*
\begin{proof}
For any strings $S$ and $T$ with $|ST| = n$, the monotonicity of $z$ under appending a suffix gives $z(TS) \geq z(T)$.
Taking the maximum over all such pairs gives $\Rot_z(n) \geq \DPre_z(n)$.
The lower bound follows from Theorem~\ref{thm:construction_prefix_deletion}, and the upper bound follows from Theorem~\ref{thm:upper_bounds}.
\end{proof}

\lzrelativeratio*
\begin{proof}
The upper bound follows from Theorem~\ref{thm:substring_complexity_bounds}.
For the lower bound, Theorem~\ref{thm:z_B_id_rho} gives $z_{\mathrm{rel}}(B_\rho) \leq z(B_\id) + z(B_\rho \mid B_\id) \in O(m)$.
Theorem~\ref{thm:z_B_rho} gives $z(B_\rho) \in \Omega(m\alpha)$.

For any sufficiently large $n$, choose the largest $\alpha$ such that $m = 2^\alpha$ satisfies $m^2 < n$.
Then $|B_\rho| = m^2 - 1 \leq n - 2$ and $m^2 < n \leq 4m^2$, so the extension lemma (Lemma~\ref{lem:extension_lemma}) gives a string $T$ of length $n$ with $z(T) = z(B_\rho) + 2 \in \Omega(m\alpha)$ and $z_{\mathrm{rel}}(T) \leq z_{\mathrm{rel}}(B_\rho) + 3 \in O(m)$.
Hence $\Ratio_{z,z_{\mathrm{rel}}}(n) \geq z(T) / z_{\mathrm{rel}}(T) \in \Omega(\log n)$ because $\alpha = \lg m \in \Theta(\log n)$.
\end{proof}

\lzcollageratio*
\begin{proof}
Since $c(T) \in O(z(T))$ for any string $T$~\cite{NavarroOP21}, there exists a constant $\beta > 0$ such that $c(T) \leq \beta z(T)$ for every nonempty string $T$.
For any string $T$ of length $n$, the symmetry of $c$ gives
\[
\Ratio_{z,c}(n)
\geq \frac{z(T^R)}{c(T^R)}
= \frac{z(T^R)}{c(T)}
\geq \frac{z(T^R)}{\beta z(T)}.
\]
Taking the maximum over all strings $T$ of length $n$ gives $\Ratio_{z,c}(n) \geq \Rev_z(n)/\beta$, so Theorem~\ref{thm:construction_string_reversal} proves the lower bound.

For the upper bound, the inequalities $\delta(T) \in O(c(T))$ and $z(T) \in O(\delta(T) \log n)$ for strings $T$ of length $n$ give $\Ratio_{z,c}(n) \in O(\log n)$.
\end{proof}
 \clearpage
\section{Cost Bounds for Height-Bounded Increasing Trees} \label{sec:increasing_tree_cost_bounds}

We first establish asymptotically tight bounds on $Q(r, h)$ and then determine $C(r, h)$ in terms of $Q(r, h)$.
For integers $h \geq 1$ and $q \geq 0$, let $B_h(q) = \binom{h + q}{h}$.

\orderEvalQ*
\begin{proof}
Let $Q = Q(r, h)$ and $R = \ln r$.
Since $r \geq 2$ and by the definition of $Q(r, h)$, we have $Q \geq 1$ and $r \leq B_h(Q)$.
We use the standard binomial coefficient bounds $(n/k)^k \leq \binom{n}{k} \leq (en/k)^k$ for integers $1 \leq k \leq n$.
Applying these bounds with $k = h$ and $k = q$ gives
\[
(1 + q/h)^h \leq B_h(q) \leq \bigl(e(1 + q/h)\bigr)^h,
\qquad
(1 + h/q)^q \leq B_h(q) \leq \bigl(e(1 + h/q)\bigr)^q
\]
for every integer $q \geq 1$.
For every $\alpha > 0$, let $\psi_\alpha(x) = x \ln(1 + \alpha/x)$ and $\Psi_\alpha(x) = x + \psi_\alpha(x)$ for $x > 0$.
Both $\Psi_\alpha$ and $\psi_\alpha$ are increasing for $x > 0$.
Taking logarithms of the binomial coefficient bounds gives
\[
\psi_h(q) \leq \ln B_h(q) \leq \Psi_h(q),
\qquad
\psi_q(h) \leq \ln B_h(q) \leq \Psi_q(h)
\]
for every integer $q \geq 1$.
Applying this bound with $q = Q$ and using $r \leq B_h(Q)$, we obtain $R = \ln r \leq \ln B_h(Q) \leq \Psi_h(Q)$.

Suppose first that $h \leq R$.
We first show that $Q > h/4$ to prove the lower bound on $Q$.
Suppose to the contrary that $Q \leq h/4$.
By the monotonicity of $\Psi_h$, we obtain
\[
R 
\leq \Psi_h(Q) \leq \Psi_h(h/4)
= \frac{h}{4}\ln\!\left(e\left(1 + \frac{h}{h/4}\right)\right)
= \frac{h}{4}\ln(5e)
< h,
\]
contradicting $h \leq R$.
Thus, $Q > h/4$.
Dividing both sides of $h < 4Q$ by $h$ gives $1 < 4Q/h$, and hence $1 + Q/h < 5Q/h$.
Using this inequality in the binomial coefficient bound gives
\[
r \leq B_h(Q) \leq \bigl(e(1 + Q/h)\bigr)^h \leq (5eQ/h)^h.
\]
Thus, $Q \geq h r^{1/h}/(5e)$.
Conversely, for $q = \lceil h r^{1/h} \rceil$, we have $r \leq (1 + q/h)^h \leq B_h(q)$, so $Q \leq q \leq 2h r^{1/h}$.
Therefore, $Q = \Theta(h r^{1/h})$.

Now suppose that $h > R$, and let $d = \ln(h/R) > 0$ and $s = R/(1 + d)$.
Using $s = R/(1 + d)$ and $h/R = e^d$, we obtain
$h/s = \frac{h}{R/(1+d)} = \frac{h}{R}(1+d) = e^d (1+d)$.
We show that $Q \in \Theta(1 + s) = \Theta\left(1 + R/\ln(1 + h/R)\right)$.
For the lower bound, suppose that $Q \leq s/8$.
By the monotonicity of $\Psi_h$, we obtain
\begin{align*}
R \leq \Psi_h(Q) 
&\leq \Psi_h(s/8) \\
&= \frac{s}{8}\ln\!\left(e\left(1 + \frac{h}{s/8}\right)\right) \\
&= \frac{s}{8}\bigl(1 + \ln(1 + 8h/s)\bigr) \\
&= \frac{R}{8(1 + d)}\bigl(1 + \ln(1 + 8e^d(1 + d))\bigr).
\end{align*}
Since $e^d(1 + d) \geq 1$, we have $1 + 8e^d(1 + d) \leq 9e^d(1 + d)$.
Taking logarithms and using $\ln(1 + d) \leq d$ gives
\[
\ln(1 + 8e^d(1 + d))
\leq \ln\bigl(9e^d(1 + d)\bigr)
= \ln 9 + d + \ln(1 + d)
\leq \ln 9 + 2d.
\]
Consequently, $R \leq R(1 + \ln 9 + 2d)/(8(1 + d)) < R$, a contradiction.
Thus, $Q > s/8$, and together with $Q \geq 1$ this gives $Q \geq (1 + s)/9$.

For the upper bound, let $q = \lceil 4s \rceil$.
If $d \geq 1$, then $(1 + d)/2 \geq 1$.
Since $q \geq 2s$, the monotonicity of $\psi_h$ gives
\[
\ln B_h(q)
\geq \psi_h(q) \geq \psi_h(2s)
= 2s \ln\!\left(1 + \frac{h}{2s}\right)
= \frac{2R}{1 + d}\ln\!\left(1 + \frac{e^d(1 + d)}{2}\right)
\geq \frac{2Rd}{1 + d}
\geq R.
\]
If $0 < d < 1$, then $q = \lceil 4s \rceil \geq 4R/(1 + d) > 2R$.
Since $h > R$, the monotonicity of $\psi_q$ gives $\ln B_h(q) \geq \psi_q(h) \geq \psi_q(R) = R \ln(1 + q/R) > R$.
Both cases give $\ln B_h(q) \geq R$, so $B_h(q) \geq r$.
By the definition of $Q$, we obtain $Q \leq q = \lceil 4s \rceil \leq 4(1 + s)$.
By combining this upper bound with the lower bound $Q \geq (1 + s)/9$, we obtain $(1 + s)/9 \leq Q \leq 4(1 + s)$.
Consequently, $Q \in \Theta(1 + s) = \Theta\left(1 + R/\ln(1 + h/R)\right)$.
\end{proof}

We now determine the exact cost $C(r, h)$ in terms of $Q(r, h)$.
The following recurrence expresses $C(r, h)$ in terms of the minimum costs of two smaller increasing trees.
\begin{lemma} \label{lem:recursive_tree_recurrence}
For any positive integer $h$ and any integer $r \geq 2$,
\[
C(r, h) = \min_{1 \leq k < r} \left( C(k, h) + k + C(r - k, h - 1) \right).
\]
\end{lemma}
\begin{proof}
Let $\calT$ be an increasing tree on $r$ vertices of height at most $h$ such that $\cost(\calT) = C(r, h)$.
Let $k$ be the largest child of the root $0$ in $\calT$.
We first transform $\calT$ so that deleting the edge between $0$ and $k$ separates the vertex sets $\{0,\ldots,k-1\}$ and $\{k,\ldots,r-1\}$.

Let $J = \{j \in \Z_r \mid j>k \text{ and } p_{\calT}(j)<k\}$.
The maximality of $k$ implies that $p_{\calT}(j) \neq 0$ for every $j \in J$.
Let $\calT'$ be the increasing tree defined by
\[
p_{\calT'}(x) =
\begin{cases}
k & \text{if } x \in J, \\
p_{\calT}(x) & \text{if } x \in \{1, \dots, r - 1\} \setminus J.
\end{cases}
\]
For every $j \in J$, the edge to its parent has length $j-p_{\calT}(j)$ in $\calT$ and $j-k$ in $\calT'$.
Since $p_{\calT}(j)<k$, these lengths satisfy $j-k \leq j-p_{\calT}(j)$.
All other edges have the same length in $\calT$ and $\calT'$, so $\cost(\calT') \leq \cost(\calT)$.
For every $j \in J$, we have $\depth_{\calT'}(j) = \depth_{\calT'}(k) + 1 = 2$, while $\depth_{\calT}(j) = \depth_{\calT}(p_{\calT}(j)) + 1 \geq 2$.
The latter inequality follows from $p_{\calT}(j) \neq 0$.
The subtrees rooted at the vertices in $J$ are pairwise disjoint, and all parent relations within each subtree are the same in $\calT$ and $\calT'$.
Hence, the depth of every vertex in these subtrees does not increase.
Every vertex outside these subtrees has the same path to the root in both trees, so its depth also does not increase.
Thus, $\depth_{\calT'}(x) \leq \depth_{\calT}(x)$ for every $x \in \Z_r$.
Overall, the cost of $\calT'$ is at most that of $\calT$, and the depth of each vertex in $\calT'$ is at most its depth in $\calT$.

There are no edges in $\calT'$ between $\{0,\ldots,k-1\}$ and $\{k,\ldots,r-1\}$ except for the edge between $0$ and $k$.
Deleting this edge therefore gives an increasing tree $\calT_L$ on $\{0,\ldots,k-1\}$ and an increasing tree $\calT_R$ rooted at $k$ on $\{k,\ldots,r-1\}$.
Relabeling each vertex $x$ of $\calT_R$ by $x-k$ gives an increasing tree on $r-k$ vertices rooted at $0$.
Since every vertex of $\calT_L$ has the same depth as in $\calT'$ and every vertex of $\calT_R$ has depth one less than in $\calT'$, $\height(\calT_L) \leq h$ and $\height(\calT_R) \leq h-1$.
The deleted edge has length $k$, so the decomposition into $\calT_L$ and $\calT_R$ gives
\begin{align*}
C(r, h)
  &= \cost(\calT) \\
  &\geq \cost(\calT') \\
  &= \cost(\calT_L) + k + \cost(\calT_R) \\
  &\geq C(k, h) + k + C(r - k, h - 1) \\
  &\geq \min_{1 \leq k' < r} \left(C(k', h) + k' + C(r - k', h - 1)\right).
\end{align*}

For the reverse inequality, let
$k \in \argmin_{1 \leq k' < r} \left(C(k', h) + k' + C(r - k', h - 1)\right)$.
Let $\calT_L$ be a minimum-cost increasing tree on $k$ vertices of height at most $h$,
and let $\calT_R$ be a minimum-cost increasing tree on $r-k$ vertices of height at most $h-1$.
By adding $k$ to every vertex label of $\calT_R$ and then joining the roots $0$ and $k$ by an edge, we combine $\calT_L$ and $\calT_R$.
The resulting tree is an increasing tree on $r$ vertices of height at most $h$ and has cost
$C(k, h) + k + C(r - k, h - 1)$, which proves the claimed equality.
\end{proof}

We have $C(1, h) = 0$ for every $h \geq 1$ and $C(r, 1) = r(r - 1)/2$ for every $r \geq 1$.
For $C(r, 1)$, this follows because the only increasing tree of height at most $1$ has every nonroot vertex as a child of the root.
Griewank and Walther~\cite{GriewankW00} gave an exact formula for the function defined by these initial values and the recurrence in Lemma~\ref{lem:recursive_tree_recurrence}.
Their result gives the following expression for $C(r, h)$.

\begin{theorem}[Griewank and Walther~{\cite[Proposition~1]{GriewankW00}}] \label{thm:checkpointing_exact_cost}
For all positive integers $r$ and $h$,
\[
C(r, h) = r Q(r, h) - \binom{h + Q(r, h)}{h + 1}.
\]
\end{theorem}

The exact formula in Theorem~\ref{thm:checkpointing_exact_cost} gives the following bounds.
\increasingtreeexactcost*
\begin{proof}
Let $Q = Q(r, h)$.
By the definition of $Q$, we have $B_h(Q - 1) < r \leq B_h(Q)$.
By Theorem~\ref{thm:checkpointing_exact_cost} and the hockey-stick identity,
\[
C(r, h) = r Q - \binom{h + Q}{h + 1} = r Q - \sum_{q = 0}^{Q - 1} B_h(q).
\]

Since $B_h(0) = 1 < r$, we have $Q \geq 1$.
Since $B_h(q + 1) - B_h(q) = \binom{h + q}{h - 1}$ is nondecreasing in $q$, we have $B_h(q) - B_h(0) \leq B_h(Q - 1) - B_h(Q - 1 - q)$ for every $0 \leq q < Q$.
Thus, $B_h(q) + B_h(Q - 1 - q) \leq B_h(0) + B_h(Q - 1) \leq r$, where the last inequality follows from $B_h(0) = 1$ and $B_h(Q - 1) < r$.
Summing over $q$ and dividing by $2$ gives
\[
\sum_{q = 0}^{Q - 1} B_h(q)
= \frac{1}{2} \sum_{q = 0}^{Q - 1} \bigl(B_h(q) + B_h(Q - 1 - q)\bigr)
\leq \frac{1}{2} r Q.
\]
Since $C(r, h) = r Q - \sum_{q = 0}^{Q - 1} B_h(q)$, we obtain $C(r, h) \geq \frac{1}{2} r Q$.
Moreover, $\sum_{q = 0}^{Q - 1} B_h(q) \geq B_h(0) = 1$, so $C(r, h) = r Q - \sum_{q = 0}^{Q - 1} B_h(q) < r Q$.
\end{proof}
 \clearpage
\section{Upper Bounds for Height-Bounded LZ-like Encodings} \label{sec:height_bounded_upper_bounds}

In this section, we prove an upper bound for height-bounded LZ-like encodings that matches our lower bound for $h \geq \lg n$.
Specifically, we give a proof of the following theorem.
\generalstringheightupperbound*

\subsection{Basic Notation and Properties}

To prove the above theorem, we consider a fixed nonempty string $S$ of length $n$ with $2z(S) \leq n$,
and let $z = z(S) + 2$, $d = \lceil \log_2((n + 2)/z) \rceil$, and $N = z2^d$.
Then, $d \geq 1$ and $n + 2 \leq N \leq 2(n + 2)$.
If $2z(S) > n$, then the LZ-like encoding of $S$ consisting only of literal phrases satisfies the height and size bounds in Theorem~\ref{thm:general_string_height_upper_bound}.
From the string $S$, we construct a padded string $T = S\#^{N - n}$ of length $N$, where $\#$ is a special character that does not occur in $S$.
Since the padding contains at least two characters, it adds exactly two phrases to the LZ factorization,
namely a literal phrase representing $\#$ and a single copy phrase for the remaining occurrences.
Thus, $z(T) = z$ and $|T| = N = z2^d$.
In the following, we consider the padded string $T$ instead of $S$.

\paragraph*{Block Trees}
We first construct a \emph{block tree} for $T$.
For each depth $0 \leq i \leq d$, let $\ell_i = 2^{d-i}$.
We define the $j$-th block at depth $i$ as the interval $B_{i,j} = [j\ell_i, (j+1)\ell_i - 1]$ for $0 \leq j < N/\ell_i$, with block indices starting at $0$.
Each block $B_{i,j}$ represents the substring $T[j\ell_i..(j+1)\ell_i - 1]$.
The $z$ blocks $B_{0,0}, \ldots, B_{0,z-1}$ at depth $0$ have length $2^d$ and no parent.
For $0 \leq i < d$, the left and right children of $B_{i,j}$ are $B_{i+1,2j}$ and $B_{i+1,2j+1}$, respectively.
The parent of $B_{i,j}$ at depth $i \geq 1$ is $B_{i-1,\lfloor j/2 \rfloor}$.
The leaves are the single-character blocks at depth $d$.
The block tree therefore consists of $z$ binary trees.

For each $0 < i \leq d$, let $\calC_i$ be the set of blocks belonging to a pair of adjacent blocks $B_{i,j}, B_{i,j+1}$ such that the occurrence $T[j\ell_i..(j+2)\ell_i - 1]$ corresponding to $B_{i,j} B_{i,j+1}$ is the leftmost occurrence of that substring in $T$.
Both blocks of every such pair belong to $\calC_i$.
We also define $\calC_0 = \{B_{0,j} \mid 0 \leq j < z\}$.
Each element in $\calC_i$ is called a \emph{selected block}.
Figure~\ref{fig:selected_blocks} illustrates an example of a block tree and its selected blocks.

\begin{figure}[tb]
\centering
\includegraphics[width=\linewidth]{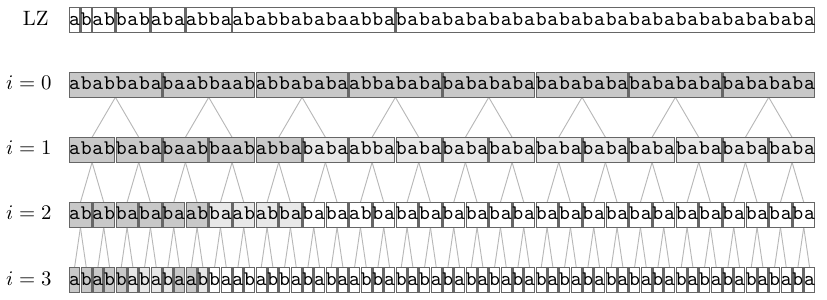}
\caption{An illustration of a block tree and its selected blocks for $N = 64$, $z = 8$, and $d = 3$.
The top row shows the LZ factorization of the string.
In the block tree, dark gray boxes represent the selected blocks, and light gray boxes at depth $i+1$ cover $U_{i,i+1}$ for $0 \leq i < d$.
}
\label{fig:selected_blocks}
\end{figure}

\begin{lemma} \label{lem:selected_block_count}
For every $0 \leq i \leq d$, $|\calC_i| \leq 3z$.
Moreover, if $B_{i, j}$ is a selected block with $i \geq 1$, then its parent $B_{i-1,\lfloor j/2 \rfloor}$ is also a selected block.
\end{lemma}
\begin{proof}
Fix $0 < i \leq d$, and suppose that $T[j\ell_i..(j+2)\ell_i - 1]$, represented by the adjacent blocks $B_{i,j}$ and $B_{i,j+1}$, is a leftmost occurrence.
If this interval were contained in a copy phrase, the same substring would occur earlier in the source of that phrase.
This contradicts the assumption that the occurrence is leftmost, so the interval contains a starting position of an LZ phrase.
Thus, for each selected block, either the block itself or one of its neighbors contains a starting position of an LZ phrase.
There are at most $z$ blocks containing such a position, and these blocks together with their neighbors cover all selected blocks, so $|\calC_i| \leq 3z$.
For $i = 0$, we have $|\calC_0| = z$ by definition.

Let $B_{i,j} \in \calC_i$ be a selected block with $i \geq 1$.
If $i = 1$, its parent is a selected block by the definition of $\calC_0$.
Otherwise, one of the adjacent pairs $(B_{i,j-1}, B_{i,j})$ and $(B_{i,j}, B_{i,j+1})$ represents a leftmost occurrence.
The interval covered by these two blocks is contained in the interval covered by $B_{i-1,\lfloor j/2 \rfloor}$ and one of its adjacent blocks at depth $i-1$.
Since the two blocks at depth $i$ represent a leftmost occurrence, the two blocks at depth $i-1$ also represent a leftmost occurrence.
Hence, $B_{i-1,\lfloor j/2 \rfloor}$ is a selected block.
\end{proof}

For each integer $0 \leq i \leq d$, let $D_i = \bigcup_{B_{i,j} \in \calC_i} B_{i,j}$ be the set of positions covered by the selected blocks at depth $i$.
For integers $0 \leq a < b \leq d$, let $U_{a,b} = D_a \setminus D_b$.

\begin{lemma} \label{lem:leftmost_occurrence_in_selected_blocks}
For every integer $0 \leq i \leq d$, if $T[l..r]$ is a leftmost occurrence and $r-l+1 \leq \ell_i$, then $[l, r] \subseteq D_i$.
\end{lemma}
\begin{proof}
The claim for $i = 0$ follows from $D_0 = [0, N-1]$.
Fix $0 < i \leq d$ and a leftmost occurrence $T[l..r]$ with $r-l+1 \leq \ell_i$.
The interval $[l, r]$ intersects at most two blocks at depth $i$.
If it intersects only one block, choose an adjacent block as well.
Let $[p, q]$ be the interval covered by these two adjacent blocks.
Since $[l, r] \subseteq [p, q]$ and $T[l..r]$ is a leftmost occurrence, $T[p..q]$ is also a leftmost occurrence.
Both blocks are therefore selected blocks, so $[l, r] \subseteq [p, q] \subseteq D_i$.

\end{proof}

\paragraph*{Partial LZ-like Encodings}
A \emph{partial LZ-like encoding} $\calE$ of $T$ consists of a set of copy phrases whose intervals in $T$ are pairwise disjoint.
For a set $\mathcal{S} \subseteq [0, N-1]$ of positions in $T$, we say that $\calE$ is a partial LZ-like encoding \emph{on} $\mathcal{S}$ if $\mathcal{S}$ is the disjoint union of the intervals of its copy phrases.
We define the size $|\calE|$ as the number of its copy phrases.
The transition function and height of $\calE$ are those of the LZ-like encoding obtained by adding a literal phrase at every position outside $\mathcal{S}$.
Thus, $\tau_{\calE}(x) = \bot$ if and only if $x \notin \mathcal{S}$ for every $x \in [0, N-1]$.

Let $\calE_1$ and $\calE_2$ be partial LZ-like encodings of $T$ on disjoint sets $\mathcal{S}_1$ and $\mathcal{S}_2$, respectively.
We define their \emph{union} $\calE_1 \sqcup \calE_2$ as the partial LZ-like encoding on $\mathcal{S}_1 \cup \mathcal{S}_2$ consisting of all copy phrases in $\calE_1$ and $\calE_2$, with size $|\calE_1| + |\calE_2|$.
By definition, we have $\tau_{\calE_1 \sqcup \calE_2}(x) = \tau_{\calE_k}(x)$ if $x \in \mathcal{S}_k$ for some $k \in \{1, 2\}$, and $\tau_{\calE_1 \sqcup \calE_2}(x) = \bot$ otherwise.

Let $\calE_1$ and $\calE_2$ be partial LZ-like encodings of $T$ on $\mathcal{S}_1$ and $\mathcal{S}_2$, respectively.
Suppose that the interval of each copy phrase in $\calE_1$ is contained in the interval of a copy phrase in $\calE_2$.
We define the \emph{merge} of $\calE_1$ into $\calE_2$ as the partial LZ-like encoding $\calE_1 \triangleright \calE_2$ consisting of all copy phrases in $\calE_1$ and the copy phrases specified below.
For each copy phrase in $\calE_2$ with interval $[l, r]$ and source position $s$, we partition $[l, r] \setminus \mathcal{S}_1$ into maximal nonempty intervals.
For each resulting interval $[p, q]$, we include a copy phrase for $T[p..q]$ with source position $s + (p - l)$ in $\calE_1 \triangleright \calE_2$.
This source position is valid because $s+p-l < p$ and $T[p..q] = T[s+p-l..s+q-l]$.
The interval of each copy phrase in $\calE_1$ adds at most two boundaries within the interval of a copy phrase in $\calE_2$, so $|\calE_1 \triangleright \calE_2| \leq 2|\calE_1| + |\calE_2|$.

\subsection{Constructing Partial LZ-like Encodings}

In this subsection, we prove the following theorem.
\begin{theorem} \label{thm:partial_encoding_between_depths}
For any integers $0 \leq a < b \leq d$, there exists a partial LZ-like encoding $\calE$ of $T$ on $U_{a,b}$ with $|\calE| \leq 52z$ such that every position in $U_{a,b}$ reaches $D_b$ within $2^{b-a}$ transitions.
That is, for every $p \in U_{a,b}$, there exists an integer $1 \leq h' \leq 2^{b-a}$ such that $\tau_{\calE}^{h'}(p) \in D_b$.
\end{theorem}
To prove this theorem, we first construct a partial LZ-like encoding $\calE_{\mathrm{init}}$~(Section~\ref{sse:initial_construction}) and then modify it to satisfy the requirements of Theorem~\ref{thm:partial_encoding_between_depths}~(Section~\ref{sse:modification}).

\subsubsection{Initial Construction} \label{sse:initial_construction}

Fix integers $0 \leq a < b \leq d$, and let $U = U_{a,b}$.
For each block $B \in \calC_a$, we partition $B \setminus D_b$ into maximal nonempty intervals and let $\mathcal{U}$ be the set of all resulting intervals.
If $B \in \calC_a$ contains $k$ blocks in $\calC_b$, then $B \setminus D_b$ consists of at most $k+1$ maximal nonempty intervals.
Hence, $|\mathcal{U}| \leq |\calC_a| + |\calC_b|$.
We further partition each interval in $\mathcal{U}$ at the phrase boundaries in the LZ factorization of $T$ and denote the resulting set of intervals by $\mathcal{U}'$.
Let $\mathcal{Z}$ be the partition of $[0, N-1]$ into the $z$ intervals corresponding to the phrases in this factorization.
Then, $\mathcal{U}' = \{I \cap J \mid I \in \mathcal{U}, J \in \mathcal{Z}, I \cap J \neq \varnothing\}$.
Since the factorization has $z-1$ boundaries, $|\mathcal{U}'| \leq |\mathcal{U}| + z-1 \leq |\calC_a| + |\calC_b| + z-1 \leq 7z-1$.
Each interval in $\mathcal{U}'$ is contained in one block at depth $a$, so it has length at most $\ell_a$.
Figure~\ref{fig:uprime_partition} illustrates the partition $\mathcal{U}'$ and the resulting partial LZ-like encoding $\calE_{\mathrm{init}}$.

\begin{figure}[tb]
\centering
\includegraphics[width=\linewidth]{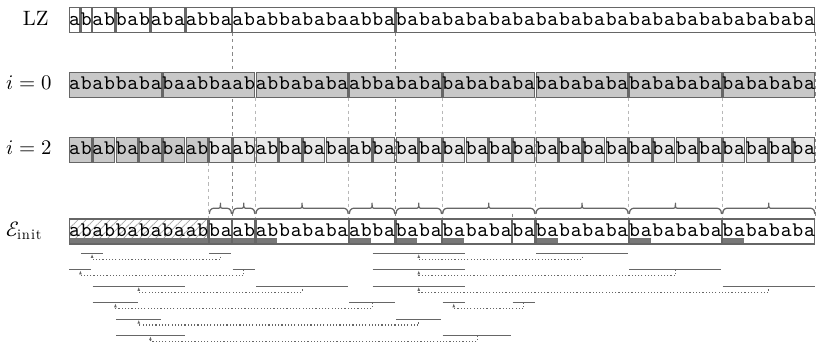}
\caption{An illustration of the partial LZ-like encoding $\calE_{\mathrm{init}}$, the partition $\mathcal{U}'$, and the set $\calX_0$ for $a = 0$ and $b = 2$ using the string in Figure~\ref{fig:selected_blocks}.
Each dark gray box represents a selected block, and each light gray box at depth $2$ represents a block contained in $U_{0,2}$.
The row labeled $\calE_{\mathrm{init}}$ shows $D_2$ as a hatched box followed by the ten copy phrases on $U_{0,2}$.
Upward braces above this row represent the intervals in $\mathcal{U}'$.
The arrows represent the reference relations of the copy phrases in $\calE_{\mathrm{init}}$ to their source occurrences.
The underlined characters represent the positions in $\calX_0$.
}
\label{fig:uprime_partition}
\end{figure}

We construct a partial LZ-like encoding $\calE_{\mathrm{init}}$ of $T$ on $U$ by dividing each interval $[l, r] \in \mathcal{U}'$ into one or two copy phrases.
In this initial construction, we do not guarantee that $\height(\calE_{\mathrm{init}}) \leq 2^{b-a}$.
We first show that $T[l..r]$ is not a leftmost occurrence.
By Lemma~\ref{lem:leftmost_occurrence_in_selected_blocks}, every literal phrase in the LZ factorization of $T$ lies in $D_b$.
Since $[l, r]$ is contained in a single LZ phrase and is disjoint from $D_b$, this phrase is a copy phrase.
Thus, its source contains an earlier occurrence of $T[l..r]$.
Let $T[l'..r']$ be the leftmost occurrence of $T[l..r]$, and let $\delta = l-l' > 0$.
If $[l, r]$ and $[l', r']$ do not overlap, we create one copy phrase for $T[l..r]$ referencing $T[l'..r']$.
If the two intervals overlap, then $\delta < r-l+1$, and we create two copy phrases for the intervals $[l, l+\delta-1]$ and $[l+\delta, r]$.
For the first phrase, we choose the starting position $s$ of the leftmost occurrence of $T[l..l+\delta-1]$ as its source position.
Since $s \leq l' = l-\delta$, its source interval ends before $l$.
The equality $T[l..r] = T[l'..r']$ implies that $T[l..r]$ has period $\delta$, so we assign source position $l$ to the second phrase.
By the definition of the transition function, each position in the second phrase refers to a position in the first phrase.
Let $\calR$ be the set of second phrases created when the two intervals overlap, and let $\calH$ be the set of all remaining phrases constructed above.
We define $\calE_{\mathrm{init}}$ as the partial LZ-like encoding on $U$ consisting of $\calH \cup \calR$.
For a set of phrases $\mathcal{A}$, we say that $x$ is a position in $\mathcal{A}$ if $x$ is contained in a phrase $I \in \mathcal{A}$.

We will use the following property of $\calE_{\mathrm{init}}$ to bound the number of transitions in the final partial LZ-like encoding.
\begin{property} \label{prop:initial_encoding_properties}
Every position in a phrase in $\calH$ reaches a position in a different phrase in $\calH$ or a position in $D_b$ after at most two transitions under $\tau_{\calE_{\mathrm{init}}}$.
\end{property}
\begin{proof}
Each phrase in $\calH$ has length at most $\ell_a$ and refers to its leftmost occurrence.
Its source interval is contained in $D_a = U \sqcup D_b$ by Lemma~\ref{lem:leftmost_occurrence_in_selected_blocks}.
By construction, each position in $\calR$ refers to a position in the immediately preceding phrase in $\calH$.
A position in $\calH$ that refers to a position in $\calR$ reaches a position in $\calH$ after one further transition.
The source interval of a phrase in $\calH$ ends before the phrase starts, so the next position in $\calH$ belongs to a different phrase.
\end{proof}

\subsubsection{Modifying the Partial LZ-Like Encoding} \label{sse:modification}

In the following, we transform $\calE_{\mathrm{init}}$ so that every position in $U$ reaches $D_b$ within $2^{b-a}$ transitions.
Let $\ell = \ell_b$.
We call a phrase \emph{long} if its length is greater than $\ell$ and \emph{short} otherwise.

For each $x \in U$, we define $\phi(x) = \tau_{\calE_{\mathrm{init}}}^{h_x}(x)$, where $h_x$ is the smallest positive integer such that $\tau_{\calE_{\mathrm{init}}}^{h_x}(x)$ is a position in $\calH$ or $D_b$.
For $x \in D_b$, we set $\phi(x) = x$, so $\phi$ is a function from $D_a$ to $D_a$.
By Property~\ref{prop:initial_encoding_properties} and the construction of $\calR$, every $x \in U$ reaches $\phi(x)$ after at most two transitions and hence $h_x \in \{1, 2\}$.

Let $k = 2^{b-a-1}-1$.
We define $\calX_0$ as the union of $D_b$, the intervals of all short phrases in $\calH$, and the first $\ell$ positions of each long phrase in $\calH$.
We perform the following steps at each iteration $t \geq 0$.
\begin{enumerate}
\item Let $\calM_t = \{x \mid x \text{ is a position in } \calH,\ \phi^p(x) \notin \calX_t \text{ for all } p \in \{0, \ldots, k\}\}$.
If $\calM_t = \varnothing$, stop the procedure.
\item Choose $x_t \in \argmax_{x \in \calM_t} \phi^k(x)$ and let $y_t = \phi^k(x_t)$.
Let $[l_t, r_t]$ be the interval of the long phrase in $\calH$ containing $y_t$.
\item Let $J_t = [y_t - \ell + 1, \min\{r_t, y_t + \ell - 1\}]$ and $\calX_{t+1} = \calX_t \cup J_t$.
\end{enumerate}
Since $y_t \notin \calX_0$ and $\calX_0$ contains the first $\ell$ positions of every long phrase in $\calH$, we have $J_t \subseteq [l_t, r_t]$ and $|J_t| \leq 2\ell-1$.
Each iteration adds at least one new position to $\calX_{t+1}$, so the procedure terminates.
At termination, every position $x$ in $\calH$ satisfies $\phi^p(x) \in \calX_t$ for some integer $0 \leq p \leq k$.
By Property~\ref{prop:initial_encoding_properties}, every position in $\calH$ reaches a position in $\calX_t$ after at most $2k$ transitions under $\tau_{\calE_{\mathrm{init}}}$.

Before bounding the number of iterations, we show how to use the sets $\calX_t$ to construct a partial LZ-like encoding with the required bound on the number of transitions.
\begin{lemma} \label{lem:encoding_from_marked_intervals}
Let $\tmax$ be the total number of iterations in the above procedure.
There exists a partial LZ-like encoding $\calE$ of $T$ on $U_{a,b}$ with $|\calE| \leq 28z + 4\tmax$ such that every position in $U_{a,b}$ reaches $D_b$ within $2^{b-a}$ transitions.
\end{lemma}
\begin{proof}
We will construct a partial LZ-like encoding from a partition of $\calX_{\tmax} \setminus D_b$ into intervals of length at most $\ell$, each contained in a single phrase in $\calH$.
For each phrase $H \in \calH$, we partition each maximal interval of $H \cap (\calX_{\tmax} \setminus D_b)$ into intervals of length $\ell$, with at most one interval of length less than $\ell$.
Let $\mathcal{I}'$ be the set of all resulting intervals.

Let $q_H$ be the number of intervals $J_t$ contained in $H$.
The set $H \cap (\calX_{\tmax} \setminus D_b)$ is the union of the prefix of $H$ of length $\min\{|H|, \ell\}$ and all intervals $J_t$ contained in $H$.
Since each $J_t$ has length at most $2\ell$, the number of intervals in $\mathcal{I}'$ contained in $H$ is at most $1 + 2q_H$.
Since $\sum_{H \in \calH} q_H = \tmax$, we have $|\mathcal{I}'| \leq |\calH| + 2\tmax$.

For each $[u, v] \in \mathcal{I}'$, we create a copy phrase representing $T[u..v]$ and choose the starting position of its leftmost occurrence as the source position.
Since $v-u+1 \leq \ell = \ell_b$, its source interval is contained in $D_b$ by Lemma~\ref{lem:leftmost_occurrence_in_selected_blocks}.
This occurrence starts before $u$ because $[u, v] \subseteq U$ is disjoint from $D_b$.
Let $\calE_1$ be the partial LZ-like encoding on $\calX_{\tmax} \setminus D_b$ consisting of these phrases.
We define $\calE = \calE_1 \triangleright \calE_{\mathrm{init}}$.
The merge is valid because the interval of every phrase in $\calE_1$ is contained in the interval of a phrase in $\calH \subseteq \calE_{\mathrm{init}}$.
The construction of $\calE_{\mathrm{init}}$ gives $|\calR| \leq |\calH| = |\mathcal{U}'| \leq 7z-1$, so the size bound for the merge gives
\[
\begin{aligned}
|\calE|
&\leq 2|\calE_1| + |\calE_{\mathrm{init}}| \\
&= 2|\mathcal{I}'| + |\calH| + |\calR| \\
&\leq 2(|\calH| + 2\tmax) + |\calH| + |\calR| \\
&\leq 4|\mathcal{U}'| + 4\tmax \\
&\leq 28z + 4\tmax.
\end{aligned}
\]

We next bound the number of transitions to $D_b$.
Every position in $\calX_{\tmax} \setminus D_b$ refers to a position in $D_b$ under $\calE$.
Since the interval of each phrase in $\calH$ is disjoint from its source interval, the transitions from positions in $\calH$ outside $\calX_{\tmax}$ are the same under $\calE$ and $\calE_{\mathrm{init}}$.
Every phrase in $\calR$ also belongs to $\calE$, so $\tau_{\calE}(x) = \tau_{\calE_{\mathrm{init}}}(x)$ for every $x \in U \setminus \calX_{\tmax}$.

Fix $x \in U$, and let $h_x$ be the smallest nonnegative integer such that $\tau_{\calE_{\mathrm{init}}}^{h_x}(x) \in \calX_{\tmax}$.
The stopping condition and Property~\ref{prop:initial_encoding_properties} give $h_x \leq 2k$ if $x$ is a position in $\calH$.
If $x$ is a position in $\calR$, its first transition reaches a position in $\calH$, so $h_x \leq 2k+1$.
Since $\tau_{\calE}(w) = \tau_{\calE_{\mathrm{init}}}(w)$ for every $w \in U \setminus \calX_{\tmax}$, the definition of $h_x$ gives $\tau_{\calE}^{h_x}(x) = \tau_{\calE_{\mathrm{init}}}^{h_x}(x)$.
If this position is not in $D_b$, one further transition under $\tau_{\calE}$ reaches $D_b$.
Thus, every position in $U$ reaches $D_b$ within $2k+2 = 2^{b-a}$ transitions under $\tau_{\calE}$.
\end{proof}

We now bound the total number of iterations $\tmax$.
We first prove the following lemma.
\begin{lemma} \label{lem:f_preserves_offsets}
Let $x-1$ and $x$ be positions in the same phrase in $\calH$.
Suppose that $\phi(x)$ lies in a phrase in $\calH$ and is not the first position of that phrase.
Then $\phi(x-1) = \phi(x)-1$.
\end{lemma}
\begin{proof}
Let $v = \tau_{\calE_{\mathrm{init}}}(x)$.
The phrase containing $x$ has a source interval disjoint from its own interval, so $\tau_{\calE_{\mathrm{init}}}(x-1) = v-1$.
If $v$ is a position in $\calH$, then $v = \phi(x)$ is not the first position of its phrase, so $\phi(x-1) = v-1 = \phi(x)-1$.
If $v$ is a position in $\calR$, the remainder in the definition of $\tau_{\calE_{\mathrm{init}}}(v)$ is positive because $\phi(x)$ is not the first position of its phrase.
Thus, $v-1$ is in the same phrase in $\calR$ as $v$, and $\phi(x-1) = \tau_{\calE_{\mathrm{init}}}(v-1) = \tau_{\calE_{\mathrm{init}}}(v)-1 = \phi(x)-1$.
\end{proof}

For integers $0 \leq t < \tmax$ and $0 \leq p \leq k$, let $I_{t,p} = [\phi^p(x_t)-\ell+1, \phi^p(x_t)]$.
Since $\phi^p(x_t) \notin \calX_0$, the interval $I_{t,p}$ is contained in the interval of the long phrase in $\calH$ containing $\phi^p(x_t)$.
Thus, $I_{t,p} \subseteq U$.
Applying Lemma~\ref{lem:f_preserves_offsets} repeatedly gives $\phi^p(x_t-q) = \phi^p(x_t)-q$ for all integers $0 \leq p \leq k$ and $0 \leq q < \ell$.
Thus, $\{\phi^{p'-p}(x) \mid x \in I_{t,p}\} = I_{t,p'}$ for all integers $0 \leq p \leq p' \leq k$.

\begin{lemma} \label{lem:iteration_intervals_disjoint}
The intervals $I_{t,p}$ are pairwise disjoint for $0 \leq t < \tmax$ and $0 \leq p \leq k$.
\end{lemma}
\begin{proof}
For a fixed $t$, each application of $\phi$ along $x_t, \phi(x_t), \ldots, \phi^k(x_t)$ reaches a different phrase in $\calH$ by Property~\ref{prop:initial_encoding_properties}.
The intervals of these phrases occur from right to left because every transition decreases the position.
Thus, $I_{t,0}, \ldots, I_{t,k}$ are pairwise disjoint.

Suppose that $I_{t',p'} \cap I_{t,p} \neq \varnothing$ for some integers $0 \leq t' < t < \tmax$ and $p, p' \in \{0, \ldots, k\}$.
Take $x \in I_{t',p'} \cap I_{t,p}$ and let $h = k - \max\{p, p'\}$.
We have $\phi^h(x) \in I_{t',p'+h} \cap I_{t,p+h}$ by the property of $I_{t,p}$ stated above.
Thus, after replacing $p$ and $p'$ by $p+h$ and $p'+h$, respectively, we may assume that $\max\{p, p'\} = k$.
Since the $I_{t,p}$ and $I_{t',p'}$ have a common element, both intervals belong to the same phrase $H \in \calH$.

If $p' = k$, the interval $I_{t',k}$ has right endpoint $y_{t'}$.
Since $I_{t,p}$ and $I_{t',k}$ have length $\ell$ and intersect, their right endpoints satisfy $|\phi^p(x_t)-y_{t'}| < \ell$.
The interval $J_{t'}$ contains all positions in the phrase $H$ at distance less than $\ell$ from $y_{t'}$.
Since $\phi^p(x_t)$ is a position in $H$, we have $\phi^p(x_t) \in J_{t'} \subseteq \calX_{t'+1} \subseteq \calX_t$.
The choice of $x_t$ gives $\phi^p(x_t) \notin \calX_t$, a contradiction.

If $p' < k$, then $p = k$ and $\phi^{p'}(x_{t'})$ belongs to the same phrase $H$ as $y_t$.
The position $x_t$ was also a candidate in iteration $t'$ because $\calX_{t'} \subseteq \calX_t$.
The position $\phi^{p'+1}(x_{t'})$ lies to the left of the interval of $H$, so $y_{t'} = \phi^k(x_{t'}) \leq \phi^{p'+1}(x_{t'}) < y_t$.
This contradicts the choice of $y_{t'}$ as the rightmost position reached by $k$ applications of $\phi$ from a candidate in iteration $t'$.
\end{proof}

We use the above lemma to bound the total number of iterations.
\begin{lemma} \label{lem:iteration_count_bound}
The total number of iterations satisfies $\tmax \leq 6z$.
\end{lemma}
\begin{proof}
For each iteration $0 \leq t < \tmax$, there are $k+1$ intervals $I_{t,0}, \ldots, I_{t,k}$, each of length $\ell_b$.
Over the $\tmax$ iterations, the sum of the lengths of all these intervals is $\tmax(k+1)\ell_b$.
Since all these intervals are contained in $U$ and are pairwise disjoint by Lemma~\ref{lem:iteration_intervals_disjoint}, we have
\[
\begin{aligned}
\tmax(k+1)\ell_b
&= \sum_{t=0}^{\tmax-1} \sum_{p=0}^{k} |I_{t,p}| \\
&\leq |U| \leq |D_a| = |\calC_a|\ell_a \leq 3z\ell_a,
\end{aligned}
\]
where the last inequality follows from Lemma~\ref{lem:selected_block_count}.
Since $k+1 = 2^{b-a-1} = \ell_a/(2\ell_b)$, we have $\tmax \leq 6z$.
\end{proof}

We now combine Lemma~\ref{lem:encoding_from_marked_intervals} and Lemma~\ref{lem:iteration_count_bound} to prove Theorem~\ref{thm:partial_encoding_between_depths}.
\begin{proof}[Proof of Theorem~\ref{thm:partial_encoding_between_depths}]
By Lemma~\ref{lem:encoding_from_marked_intervals}, there exists a partial LZ-like encoding $\calE$ of $T$ on $U_{a,b}$ with $|\calE| \leq 28z + 4\tmax$ such that every position in $U_{a,b}$ reaches $D_b$ within $2^{b-a}$ transitions.
Since $\tmax \leq 6z$ by Lemma~\ref{lem:iteration_count_bound}, we have $|\calE| \leq 28z + 4\tmax \leq 28z + 4 \cdot 6z = 52z$.
\end{proof}

\subsection{Constructing the Height-Bounded LZ-like Encoding}
We next construct an LZ-like encoding of $T$ with bounded height using partial LZ-like encodings constructed above.
We first give a construction whose size and height bounds depend on a parameter $\alpha$ and then construct a small encoding with height at most $h$ for each integer $h \geq \lg (N/z)$.

The following theorem gives the size and height bounds in terms of a parameter $\alpha$.
\begin{theorem} \label{thm:baseB_upper_bound}
For the string $T$ defined above and any integer $\alpha \geq 2$, there exists an LZ-like encoding $\calE$ of $T$ with height $\height(\calE) \leq 2\alpha\left(1 + \log_{\alpha}\frac{N}{z}\right)$ and size $|\calE| \leq 52z\log_{\alpha}(N/z) + 55z$.
\end{theorem}
\begin{proof}
For a fixed parameter $\alpha \geq 2$, 
let $g = \lceil \lg \alpha \rceil \geq 1$, $s = \lceil d / g \rceil$,
and define a sequence $(a_i)_{i=0}^s$ where $a_i = \min \{ ig, d \}$.
By the definition, $a_0 = 0$, $a_s = d$ and the sequence $(a_i)_{i=0}^s$ is strictly increasing.

We apply Theorem~\ref{thm:partial_encoding_between_depths} with depths $a_i$ and $a_{i+1}$ to obtain a partial LZ-like encoding $\calE_i$ on $U_{a_i, a_{i+1}}$ for $0 \leq i < s$.
The sets $U_{a_i, a_{i+1}}$ are pairwise disjoint subsets of $D_0 = [0, N-1]$, and their disjoint union is $D_0 \setminus D_d$, so $\bigsqcup_{0 \leq i < s}\calE_i$ is a partial LZ-like encoding on $[0, N-1] \setminus D_d$.
We then add a literal phrase at each position in $D_d$ to obtain an LZ-like encoding $\calE$ of $T$.
By summing up the number of phrases,
we obtain 
\[
|\calE| = \left(\sum_{0 \leq i < s} |\calE_i|\right) + |D_d|
\leq 52zs + \ell_d |\calC_d| \leq 52zs + 3z,\]
where the last inequality comes from Lemma~\ref{lem:selected_block_count} and $\ell_d = 1$.

We next evaluate $\height(\calE)$.
For every $0 \leq i < s$ and $x \in U_{a_i, a_{i+1}}$, we have $\tau_{\calE}(x) = \tau_{\calE_i}(x)$ by the definition of $\calE$.
By Theorem~\ref{thm:partial_encoding_between_depths}, every position in $U_{a_i, a_{i+1}}$ reaches $D_{a_{i+1}}$ within $2^{a_{i+1}-a_i} \leq 2^g \leq 2^{1 + \lg\alpha} = 2\alpha$ transitions under $\tau_{\calE}$.
Since $D_{a_i} = U_{a_i, a_{i+1}} \sqcup D_{a_{i+1}}$ for $0 \leq i < s$, every position in $D_0$ reaches $D_d$ within $\sum_{i=0}^{s-1} 2\alpha = 2\alpha s$ transitions.
Thus, $\height(\calE) \leq 2\alpha s$.

Finally, since $g \geq \lg\alpha$ and $d = \lg(N/z)$, we have $s = \lceil d/g \rceil \leq d/g + 1 \leq 1 + d/\lg\alpha = 1 + \log_{\alpha}(N/z)$.
Therefore, $|\calE| \leq 52zs + 3z \leq 52z(1 + \log_{\alpha}(N/z)) + 3z = 52z\log_{\alpha}(N/z) + 55z$.
The height satisfies $\height(\calE) \leq 2\alpha s \leq 2\alpha(1 + \log_{\alpha}(N/z))$, which completes the proof.
\end{proof}

We now use Theorem~\ref{thm:baseB_upper_bound} to construct small encodings with height at most $h$ for every integer $h \geq \lg(N/z)$.
For $h \in O(\lg(N/z))$, an encoding of size $O(z\lg(N/z))$ achieves the desired bound.
For this bound, we use the basic construction of an LZ-like encoding from a block tree.
\begin{lemma}[see also~\cite{Navarro21,Navarro17,BelazzouguiCGGK21,CicaleseGLNPU26}] \label{lem:block_tree_encoding}
There exists an LZ-like encoding $\calE$ of $T$ with $\height(\calE) \leq d$ and $|\calE| \leq 6zd$.
\end{lemma}
\begin{proof}
For each $1 \leq i \leq d$, consider the blocks of length $\ell_i$ at depth $i$ that do not belong to $\calC_i$ and whose parents belong to $\calC_{i-1}$.
These blocks form a partition of $U_{i-1,i}$.
We create one copy phrase for each of these blocks, referring to its leftmost occurrence.
By Lemma~\ref{lem:leftmost_occurrence_in_selected_blocks}, the source interval of a phrase corresponding to a block $B_{i,j}$ is contained in $D_i$ and is therefore disjoint from $B_{i,j} \subseteq U_{i-1,i}$.
The intervals in $T$ corresponding to these copy phrases form a partition of $\bigsqcup_{i=1}^{d} U_{i-1,i} = [0, N-1] \setminus D_{d}$.
We add a literal phrase at each position in $D_{d}$ to obtain an LZ-like encoding $\calE$ of $T$.

Every position in $U_{i-1,i}$ reaches $D_i$ in one transition.
Hence, every position in $D_0$ reaches $D_{d}$ within $d$ transitions, and $\height(\calE) \leq d$.
We create at most $2|\calC_{i-1}|$ copy phrases for the blocks at depth $i$, so Lemma~\ref{lem:selected_block_count}, $|\calC_0| = z$, and $\ell_d = 1$ give
\[
|\calE|
\leq \sum_{i=1}^{d} 2|\calC_{i-1}| + |D_{d}|
\leq 2z + 6z(d-1) + 3z
= 6zd-z \leq 6zd.
\]
\end{proof}

We combine Lemma~\ref{lem:block_tree_encoding} with Theorem~\ref{thm:baseB_upper_bound} to obtain the following bound.
\begin{theorem} \label{thm:prescribed_height_upper_bound}
For every integer $h \geq \lg(N/z)$, there exists an LZ-like encoding $\calE$ of $T$ with $\height(\calE) \leq h$ and $|\calE| \in O\left(z\left(1 + \frac{\lg(N/z)}{\lg(1 + h/\lg(N/z))}\right)\right)$.
\end{theorem}
\begin{proof}
Let $L = \lg(N/z) = d$.
We use Lemma~\ref{lem:block_tree_encoding} if $h \leq 8L$ and the construction in Theorem~\ref{thm:baseB_upper_bound} if $h > 8L$.

Suppose first that $L \leq h \leq 8L$.
By Lemma~\ref{lem:block_tree_encoding}, there exists an encoding $\calE$ with $\height(\calE) \leq d = L \leq h$ and $|\calE| \leq 6zd = 6zL$.
Since $1 \leq h/L \leq 8$ and $L \geq 1$, we have $L / \lg (1 + h / L) \in \Theta(L)$ and $z\left(1 + \frac{L}{\lg(1 + h/L)}\right) \in \Theta(zL)$.
Thus, $\calE$ satisfies the required size bound in this case.

Now suppose that $h > 8L$, and let $\alpha_h = \lfloor h/(4L) \rfloor \geq 2$.
We take the encoding $\calE$ constructed in the proof of Theorem~\ref{thm:baseB_upper_bound} with $\alpha = \alpha_h$.
By Theorem~\ref{thm:baseB_upper_bound} and $\lg\alpha_h \geq 1$, we obtain
\[
\begin{aligned}
\height(\calE)
&\leq 2\alpha_h\left(1 + \frac{L}{\lg\alpha_h}\right) \\
&\leq 2\alpha_h(1+L) \\
&\leq 4\alpha_h L \\
&\leq h,
\end{aligned}
\]
where the last two inequalities use $L \geq 1$ and $\alpha_h \leq h/(4L)$.

To obtain the size bound, we first bound $\lg(1+h/L)$ in terms of $\lg \alpha_h$.
Since $\alpha_h = \lfloor h/(4L) \rfloor \geq 2$, we have $h/L < 4(\alpha_h+1) = 4\alpha_h + 4 \leq 4\alpha_h + 2\alpha_h = 6\alpha_h$.
Hence, $1+h/L < 1 + 6\alpha_h \leq 2\alpha_h + 6\alpha_h = 8\alpha_h$.
By taking logarithms, we obtain $\lg(1+h/L) \leq \lg (8 \alpha_h) \leq 3 + \lg\alpha_h \leq 4\lg\alpha_h$, where the last inequality uses $\lg \alpha_h \geq 1$.
The size bound of Theorem~\ref{thm:baseB_upper_bound} now gives
\[
\begin{aligned}
|\calE|
&\leq 52z\frac{L}{\lg\alpha_h} + 55z \\
&\leq 208z\frac{L}{\lg(1+h/L)} + 55z \\
&\leq 208z\left(1 + \frac{L}{\lg(1+h/L)}\right).
\end{aligned}
\]
This proves the required size bound in both cases.
\end{proof}

We now apply the bound for the padded string $T$ to the original string $S$ to prove the main theorem in this section.
\begin{proof}[Proof of Theorem~\ref{thm:general_string_height_upper_bound}]
Let $L = \lg(n/z_S)$.
If $n \leq 2z_S$, we take a literal phrase at every position of $S$.
This encoding has height $0$ and size $n \leq 2z_S$, so it satisfies the claimed bounds.

Suppose that $n > 2z_S$, so $L > 1$.
We construct the padded string $T$ as above, with $z = z(T) = z_S + 2$, $d = \lceil \lg((n+2)/(z_S+2)) \rceil$, and $N = z2^d$.
Since $(n+2)/(z_S+2) \leq n/z_S$ and $h$ is an integer, we have $d = \lg(N/z) = \lceil \lg((n+2)/(z_S+2)) \rceil \leq \lceil \lg(n/z_S) \rceil \leq h$.
Thus, Theorem~\ref{thm:prescribed_height_upper_bound} gives an LZ-like encoding $\calE_T$ of $T$ with $\height(\calE_T) \leq h$ and $|\calE_T| \in O\left(z\left(1 + d/\lg(1 + h/d)\right)\right)$.

We next express the size bound in terms of $n$ and $z_S$.
The inequalities $z = z_S + 2 \leq 3z_S$ and $d \leq L+1 \leq 2L$ imply
\[
\lg(1 + h/d)
\geq \lg(1 + h/(2L))
= \frac{1}{2}\lg\left((1 + h/(2L))^2\right)
= \frac{1}{2}\lg\left(1 + \frac{h}{L} + \frac{h^2}{4L^2}\right)
\geq \frac{1}{2}\lg(1 + h/L).
\]
By combining the inequalities above, we obtain
\[
z\left(1 + \frac{d}{\lg(1 + h/d)}\right)
\leq 3z_S\left(1 + \frac{4L}{\lg(1 + h/L)}\right)
\leq 12z_S\left(1 + \frac{L}{\lg(1 + h/L)}\right),
\]
so $|\calE_T| \in O\left(z_S\left(1 + L/\lg(1 + h/L)\right)\right)$.

Finally, we restrict $\calE_T$ to the prefix $S = T[0..n-1]$.
We delete every phrase starting after position $n-1$ and shorten the final remaining phrase to end at position $n-1$, keeping its source position.
Since each remaining copy phrase refers to an occurrence starting to its left, its source interval is contained in $[0, n-1]$.
The resulting LZ-like encoding $\calE$ of $S$ satisfies $|\calE| \leq |\calE_T|$.
Each remaining phrase has the same starting position and source value, so $\height(\calE) \leq \height(\calE_T) \leq h$, which completes the proof.
\end{proof}

 \clearpage
\section[Equivalence between Sensitivity and Size-Separation Lower Bounds]{Equivalence between Sensitivity and\\Size-Separation Lower Bounds}
\label{sec:equivalence}

In the main body, we establish logarithmic lower bounds for the sensitivities of the LZ factorization to prefix deletion and string reversal.
We also prove logarithmic separations between the LZ factorization size $z$ and each of $c$ and $z_{\mathrm{rel}}$.
In this appendix, we show that a family of strings giving any one of these four lower bounds can be transformed into families giving the other three.
Although this equivalence is not necessary to prove our main results, it may be useful for studying relationships between sensitivity and size-separation lower bounds in other settings.

\begin{theorem}
\label{thm:equivalent}
The following four statements are equivalent.
\begin{enumerate}
    \renewcommand{\theenumi}{\roman{enumi}}
    \renewcommand{\labelenumi}{(\theenumi)}
    \item \label{enm:equiv_prefix} $\displaystyle \limsup_{n \to \infty} \frac{\DPre_z(n)}{\log n} > 0$.
    \item \label{enm:equiv_reversal} $\displaystyle \limsup_{n \to \infty} \frac{\Rev_z(n)}{\log n} > 0$.
    \item \label{enm:equiv_collage} $\displaystyle \limsup_{n \to \infty} \frac{\Ratio_{z,c}(n)}{\log n} > 0$.
    \item \label{enm:equiv_rlz} $\displaystyle \limsup_{n \to \infty} \frac{\Ratio_{z,z_{\mathrm{rel}}}(n)}{\log n} > 0$.
\end{enumerate}
\end{theorem}

These statements describe logarithmic lower bounds for the corresponding sensitivities and worst-case size ratios.
We prove the theorem by showing
\[
(\ref{enm:equiv_prefix})
\Longrightarrow (\ref{enm:equiv_rlz})
\Longrightarrow (\ref{enm:equiv_collage})
\Longrightarrow (\ref{enm:equiv_prefix})
\quad
and
\quad
(\ref{enm:equiv_collage})
\Longrightarrow (\ref{enm:equiv_reversal})
\Longrightarrow (\ref{enm:equiv_collage}).
\]

We first give a detailed definition of collage systems and introduce the additional terminology needed for the proof.
A collage system consists of nonterminals $X_1, \ldots, X_m$ with one production rule for each nonterminal.
For a nonterminal $X_k$, we denote the string derived by $X_k$ by $\exp(X_k)$ and call it the \emph{expansion} of $X_k$.
For each $1 \leq k \leq m$, the production rule for $X_k$ has one of the following four forms.
\begin{itemize}
    \item A terminal rule $X_k \rightarrow a$, where $a \in \Sigma$ and $\exp(X_k) = a$.
    \item A concatenation rule $X_k \rightarrow X_iX_j$, where $1 \leq i, j < k$ and $\exp(X_k) = \exp(X_i)\exp(X_j)$.
    \item An exponent rule $X_k \rightarrow (X_i)^r$, where $1 \leq i < k$, $r \geq 2$, and $\exp(X_k) = (\exp(X_i))^r$.
    \item A truncation rule $X_k \rightarrow X_i[b..e]$, where $1 \leq i < k$, $0 \leq b \leq e < |\exp(X_i)|$, and $\exp(X_k) = \exp(X_i)[b..e]$.
\end{itemize}
The size of a collage system is measured by the number $m$ of its nonterminals.
The string derived by the collage system is defined as $\exp(X_m)$.

A collage system deriving $T$ is \emph{internal} if $\exp(X_k)$ is a substring of $T$ for every $1 \leq k \leq m$.
We denote the minimum size of an internal collage system deriving $T$ by $c_{\mathrm{in}}(T)$.
It is known that $c_{\mathrm{in}}(T) \in \Theta(c(T))$~\cite{MigitaUI26}.

The proofs of $(\ref{enm:equiv_collage}) \Longrightarrow (\ref{enm:equiv_prefix})$ and $(\ref{enm:equiv_collage}) \Longrightarrow (\ref{enm:equiv_reversal})$ are based on the common construction given in the following lemma.

\begin{lemma}
\label{lem:collage_transfers}
For every nonempty string $T \in \Sigma^n$, there exists a string $S \in \Sigma^*$ such that $|ST| \in O(n^2)$ and $z(ST) \in O(c(T))$.
\end{lemma}

\begin{proof}
Let $X_1, \ldots, X_m$ be the nonterminals of a minimum internal collage system deriving $T$, where $m = c_{\mathrm{in}}(T) \in \Theta(c(T))$.
We define $S_0 = \varepsilon$, $S_k = \exp(X_1) \cdots \exp(X_k)$ for $1 \leq k \leq m$, and $S = S_{m-1}$.
Since $\exp(X_m) = T$, we have $ST = S_m$.

We prove by induction on $k$ that $S_k$ has an LZ-like factorization of size at most $2k$.
The claim holds for $k = 0$.
For $k \geq 1$, the induction hypothesis gives an LZ-like factorization of $S_{k-1}$ of size at most $2(k-1)$.
Since $S_k = S_{k-1}\exp(X_k)$, we construct an LZ-like factorization of $S_k$ of size at most $2k$ by appending at most two phrases representing $\exp(X_k)$ to the factorization of $S_{k-1}$, as follows.
\begin{itemize}
    \item If $X_k \rightarrow a$ is a terminal rule, then $\exp(X_k)$ is represented by one literal phrase.
    \item If $X_k \rightarrow X_iX_j$ is a concatenation rule, then both $\exp(X_i)$ and $\exp(X_j)$ occur in $S_{k-1}$.
    Their concatenation can be represented by two copy phrases.
    \item If $X_k \rightarrow (X_i)^r$ is an exponent rule, then the first copy of $\exp(X_i)$ is represented by a copy phrase with a source in $S_{k-1}$.
    The remaining part $(X_i)^{r-1}$ is represented by one copy phrase whose source starts at the beginning of that first copy.
    \item If $X_k \rightarrow X_i[b..e]$ is a truncation rule, then $\exp(X_k)$ occurs within the copy of $\exp(X_i)$ in $S_{k-1}$.
    Hence, $\exp(X_k)$ is represented by one copy phrase.
\end{itemize}
The resulting LZ-like factorization of $S_k$ has size at most $2k$, completing the induction.
In particular, $z(ST) = z(S_m) \leq 2m$, so $z(ST) \in O(c(T))$.

Since the collage system is internal, $|\exp(X_k)| \leq n$ for every $1 \leq k \leq m$.
Thus, $n \leq |ST| = \sum_{k = 1}^{m} |\exp(X_k)| \leq nm$.
We can construct an internal collage system deriving $T$ by representing each character with a terminal rule and concatenating the resulting $n$ nonterminals with $n-1$ concatenation rules,
so $m \leq 2n-1$.
Hence, $|ST| \leq nm \leq n(2n-1) \in O(n^2)$, which completes the proof.
\end{proof}

We now prove Theorem~\ref{thm:equivalent} using Lemma~\ref{lem:collage_transfers}.

\begin{proof}[Proof of Theorem~\ref{thm:equivalent}]
In each part below, we assume the statement on the left of $\Longrightarrow$ and prove the statement on the right.
We fix a constant $\alpha \geq 1$ such that the construction in Lemma~\ref{lem:collage_transfers} satisfies $z(ST) \leq \alpha c(T)$ and $n \leq |ST| \leq \alpha n^2$.
We also fix a constant $\beta \geq 1$ such that $c(T) \leq \beta z(T)$ for every nonempty string $T$.
For all sufficiently large $n$, the bound $n \leq N \leq \alpha n^2$ gives $\log N \leq 3\log n$.

\medskip
\noindent\textbf{\boldmath $(\ref{enm:equiv_prefix}) \Longrightarrow (\ref{enm:equiv_rlz})$.}
By assumption, there exist $\eta > 0$ and a sequence of pairs $(S_k, T_k)_{k \geq 1}$ such that $n_k = |S_kT_k| \geq 2$, $\lim_{k \to \infty} n_k = \infty$, and $z(T_k) \geq \eta z(S_kT_k)\log n_k$.
Since $T_k$ occurs as a suffix of $S_kT_k$, the relative LZ factorization of $T_k$ with respect to $S_kT_k$ consists of one copy phrase referring to that occurrence.
Thus, $z_{\mathrm{rel}}(T_k) \leq z(S_kT_k) + 1 \leq 2z(S_kT_k)$.
Therefore,
\[
\frac{z(T_k)}{z_{\mathrm{rel}}(T_k)}
\geq \frac{z(T_k)}{2z(S_kT_k)}
\geq \frac{\eta}{2}\log n_k
\geq \frac{\eta}{2}\log |T_k|.
\]
Since $|T_k| \geq z(T_k) \geq \eta\log n_k$, we have $\lim_{k \to \infty} |T_k| = \infty$.
Hence, the sequence $(T_k)_{k \geq 1}$ proves Statement~(\ref{enm:equiv_rlz}).

\medskip
\noindent\textbf{\boldmath $(\ref{enm:equiv_rlz}) \Longrightarrow (\ref{enm:equiv_collage})$.}
Let $T$ be a nonempty string and let $S$ be a reference string satisfying $z_{\mathrm{rel}}(T) = z(S) + z(T \mid S)$.
We can construct an LZ-like factorization of $ST$ by concatenating the LZ factorization of $S$ and the relative LZ factorization of $T$ with respect to $S$.
The resulting factorization has size $z_{\mathrm{rel}}(T)$, so $z(ST) \leq z_{\mathrm{rel}}(T)$.
Let $X_1, \ldots, X_m$ be the nonterminals of a minimum collage system deriving $ST$.
We can construct a collage system deriving $T$ by adding a nonterminal $X_{m+1}$ with the truncation rule $X_{m+1} \rightarrow X_m[|S|..|ST|-1]$, which extracts the suffix $T$.
Since $z_{\mathrm{rel}}(T) \geq 1$, we obtain
\[
c(T)
\leq c(ST) + 1
\leq \beta z(ST) + 1
\leq (\beta + 1)z_{\mathrm{rel}}(T).
\]
Thus, $\Ratio_{z,c}(n) \geq \Ratio_{z,z_{\mathrm{rel}}}(n)/(\beta + 1)$ for every $n \geq 1$.
Therefore, Statement~(\ref{enm:equiv_rlz}) implies Statement~(\ref{enm:equiv_collage}).

\medskip
\noindent\textbf{\boldmath $(\ref{enm:equiv_collage}) \Longrightarrow (\ref{enm:equiv_prefix})$.}
By assumption, there exist $\eta > 0$ and a sequence $(T_k)_{k \geq 1}$ such that $n_k = |T_k| \geq 2$, $\lim_{k \to \infty} n_k = \infty$, and $z(T_k) \geq \eta c(T_k)\log n_k$.
By Lemma~\ref{lem:collage_transfers}, there are strings $S_k$ such that $N_k = |S_kT_k|$ satisfies $n_k \leq N_k \leq \alpha n_k^2$ and $z(S_kT_k) \leq \alpha c(T_k)$.
We obtain $T_k$ by deleting the prefix $S_k$ from $S_kT_k$, so
\[
\DPre_z(N_k)
\geq \frac{z(T_k)}{z(S_kT_k)}
\geq \frac{z(T_k)}{\alpha c(T_k)}
\geq \frac{\eta}{3\alpha}\log N_k
\]
for all sufficiently large $k$.
Since $N_k \geq n_k$, we have $\lim_{k \to \infty} N_k = \infty$, which proves Statement~(\ref{enm:equiv_prefix}).

\medskip
\noindent\textbf{\boldmath $(\ref{enm:equiv_collage}) \Longrightarrow (\ref{enm:equiv_reversal})$.}
We use the same constant $\eta$ and strings $T_k$ as in the proof of $(\ref{enm:equiv_collage}) \Longrightarrow (\ref{enm:equiv_prefix})$.
By Lemma~\ref{lem:collage_transfers} and the symmetry of $c$, there are strings $S_k$ such that $n_k \leq |S_kT_k^R| \leq \alpha n_k^2$ and $z(S_kT_k^R) \leq \alpha c(T_k^R) = \alpha c(T_k)$.
Since $(S_kT_k^R)^R = T_kS_k^R$ has $T_k$ as a prefix, we have $z(T_kS_k^R) \geq z(T_k)$.
For $N_k = |S_kT_k^R|$, we therefore have
\[
\Rev_z(N_k)
\geq \frac{z(T_kS_k^R)}{z(S_kT_k^R)}
\geq \frac{z(T_k)}{\alpha c(T_k)}
\geq \frac{\eta}{3\alpha}\log N_k
\]
for all sufficiently large $k$.
Since $N_k \geq n_k$, we have $\lim_{k \to \infty} N_k = \infty$, proving Statement~(\ref{enm:equiv_reversal}).

\medskip
\noindent\textbf{\boldmath $(\ref{enm:equiv_reversal}) \Longrightarrow (\ref{enm:equiv_collage})$.}
The symmetry of $c$ and the bound $c(T) \leq \beta z(T)$ imply
\[
\frac{z(T^R)}{c(T^R)}
= \frac{z(T^R)}{c(T)}
\geq \frac{1}{\beta}\frac{z(T^R)}{z(T)}
\]
for every nonempty string $T$.
By the definitions of the two quantities, $\Ratio_{z,c}(n) \geq \Rev_z(n)/\beta$.
Therefore, Statement~(\ref{enm:equiv_reversal}) implies Statement~(\ref{enm:equiv_collage}).
Thus, the four statements are equivalent.
\end{proof}

\end{document}